\documentclass[11pt,a4paper]{article}
\pdfoutput=1
\usepackage[pdfstartview=FitH,colorlinks=true,linkcolor=black,anchorcolor=black,citecolor=black,urlcolor=black]{hyperref}
\usepackage[english]{babel}
\usepackage[top=2.7cm,right=3cm,left=3cm,bottom=2.7cm]{geometry}
\usepackage{amssymb,amsmath,amsthm,eucal,tikz,slashed,titling}
\usetikzlibrary{matrix,arrows,shapes,calc}
\usepackage{xspace}

\allowdisplaybreaks[4]
\numberwithin{equation}{section}

\newtheorem{prop}{Proposition}[section]

\newtheorem{lemma}[prop]{Lemma}

\newtheorem{df}[prop]{Definition}
\newtheorem{dfprop}[prop]{Definition/Proposition}
\newtheorem{rem}[prop]{Remark}
\newtheorem{ex}[prop]{Example}

\newcommand{\A}{\mathcal{A}}
\newcommand{\Ae}{\mathcal{A}_{\mathcal{E}}}
\newcommand{\B}{\mathcal{B}}
\newcommand{\K}{\mathcal{K}}
\newcommand{\E}{\mathcal{E}}
\renewcommand{\L}{\mathcal{L}}
\newcommand{\HH}{\mathcal{H}}
\newcommand{\HHea}{\mathcal{H}_{\mathcal{E}}}
\newcommand{\N}{\mathbb{N}}
\newcommand{\Z}{\mathbb{Z}}
\newcommand{\R}{\mathbb{R}}
\newcommand{\C}{\mathbb{C}}

\newcommand{\de}{\mathrm{d}}
\newcommand{\ket}[1]{\left|#1\right>}
\newcommand{\bra}[1]{\left<#1\right|}
\newcommand{\inner}[1]{\left<#1\right>}

\newcommand{\mat}[1]{\bigg(\!\begin{array}{cc}#1\end{array}\!\bigg)}
\newcommand{\D}{\slashed{D}}
\newcommand{\tr}{\mathrm{Tr}}
\newcommand{\nint}{\int\mkern-19mu-\;}
\newcommand{\tint}{\int\mkern-16mu-\,}
\newcommand{\mc}{\mathcal}
\newcommand{\mf}{\mathfrak}
\newcommand{\id}{\textup{\textsf{id}}}
\newcommand{\za}{\diamond}

\begin{document}

\setlength{\droptitle}{-6pc}
\pretitle{\begin{flushright}\small
ICCUB-13-220
\end{flushright}\vspace*{2pc}%
\begin{center}\LARGE}
\posttitle{\par\end{center}}

\title{Matrix Geometries Emergent from a Point}

\author{~\\ ~Francesco D'Andrea$^1$, Fedele Lizzi$\hspace{1pt}^{2,3,4}$ and Pierre Martinetti$\hspace{1pt}^{2,3}$ \\[12pt]
{\footnotesize $^1$ Dipartimento di Matematica e Applicazioni, Universit\`a di Napoli {\sl Federico II}. } \\[3pt]
{\footnotesize $^2$ Dipartimento di Fisica, Universit\`a di Napoli {\sl Federico II}.}\\
{\footnotesize $^3$ I.N.F.N. -- Sezione di Napoli.} \\[3pt]
{\footnotesize $^4$ Departament de Estructura i Constituents de la Mat\`eria.} \\
{\footnotesize Institut de Ci\'encies del Cosmos, Universitat de Barcelona.}}

\date{}

\maketitle

\begin{abstract}\noindent
We describe a categorical approach to finite noncommutative geometries.
Objects in the category are spectral triples, rather than unitary equivalence classes as in other approaches.
This enables to treat fluctuations of the metric and unitary
equivalences on the same footing, as representatives of particular morphisms in this category.
We then show how a matrix geometry (Moyal plane) emerges as a fluctuation
from one point, and discuss some geometric aspects of this space.
\end{abstract}

\medskip

%%% ======================================================================

\section{Introduction}

In the realm of $C^*$-algebras there are two important notions of equivalence, which coincide in the commutative case: $C^*$-isomorphism and strong Morita equivalence~\cite{Rie74}.
This suggests that two spaces could be considered the same when described by Morita equivalent $C^*$-algebras, a weaker requirement than $C^*$-isomorphism.
For von Neumann algebras, using the analogue of Morita-equivalence bimodules as morphisms (or ``correspondences'') is already in \cite[\S5.B]{Con94}.

This idea led, for example, to the description of an elliptic curve with modular
parameter $\tau$ ($\tau\in\C$, $\mathrm{Im}(\tau)>0$)
using the transformation group $C^*$-algebra associated to the action of $\Z+\tau\Z$ on $\C$~\cite{Mar05}.
This algebra is Morita equivalent to the algebra of continuous functions on the torus $\C/\Z+\tau\Z$,
and the complex structure (the modular parameter) is encoded in the group action. The transformation group $C^*$-algebra has the advantage
of being defined even when $\tau=2\pi\theta$ is real, giving
a $C^*$-algebra that is the suspension of the $C^*$-algebra of the noncommutative torus $C(\mathbb{T}^2_\theta)$.
One thus interprets noncommutative tori as points on the boundary of the moduli space of elliptic curves (see
e.g.~\cite{Mar05} and references therein).
Two elliptic curves are equivalent (biholomorphic) if and only if their modular
parameters are in the same $SL(2,\Z)$ orbit. Similarly two
noncommutative tori are Morita equivalent (but not isomorphic)
if and only if their $\theta$-parameters are in the same $SL(2,\Z)$
orbit~\cite{Rie81}. Hence one can argue that Morita
equivalence is a better notion than isomorphism, at least
in the context of elliptic curves (in fact, it is also believed that Morita equivalence should have
a role in particle physics, 
for instance in relation with T-duality in string theory \cite{CDS98,lls98}).

Morita equivalence is only about topology. A natural question is whether Morita equivalent
noncommutative tori are also equivalent as ``noncommutative manifolds'' (in the sense of spectral triples).
A positive answer is given in~\cite{Ven11}.

\medskip

In this paper, we define a category of spectral triples based on Morita-equivalence, by adapting some ideas
of~\cite{Ven11} and~\cite{Mes09} to the finite-dimensional case.
Following \cite{Ven11} we work with connections associated to Connes' differential calculus, rather than universal connections
as in \cite{Mes09}. Objects in our category are spectral triples,
rather than unitary equivalence classes. As a consequence, both fluctuations of the metric and unitary
equivalences are representative of particular morphisms of the category.
Then we build a spectral triple on the $1$-point space yielding a
non-trivial differential calculus, and show how to obtain the Moyal
plane as a fluctuation.
We conclude with some considerations on the metric aspect of Moyal
plane, stressing some relations between Connes spectral distance and polarization in geometric quantization.

For a review of other categorical approaches to noncommutative geometry see \cite{BCL12}.

\medskip

The paper is organized as follows. After briefly recalling basics of noncommutative geometry in
\S\ref{sec:2.1}--\ref{sec:2.3}, we discuss fluctuations of a spectral triples in \S\ref{sec:mor} and
\S\ref{sec:2.5}. ``Being a fluctuation of'' is not an equivalence relation: it
is reflexive and transitive but not symmetric~\cite[Rem.~1.143]{CM08}. Instead, inspired by
\cite{Con94,Mes09,Ven11} we think of \emph{fluctuations} as morphisms in a suitable category of
spectral triples. Transitivity and reflexivity correspond then to the existence of a composition law
of morphisms and of the identity morphisms, while the lack of symmetry means that not every morphism is an isomorphism.
In fact, to  ensure the existence of the identity morphisms, we extend the class of maps to the composition of fluctuations 
with unitary equivalences, that we call \emph{correspondences}.

Spectral triples with morphisms given by classes of correspondences
(or ``Morita morphisms'') form a category which we will call ``Morita category''.
This has two subcategories, with the same objects
but less morphisms. A first subcategory is the one with morphisms
given by unitary equivalences only, which is
a groupoid (every morphism is an isomorphism).
A second subcategory is the one with morphisms given by inner fluctuations only
(which are not always invertible).  
If we fix the algebra and the Hilbert space, objects are simply Dirac operators and morphisms are connection $1$-forms.
We call this category the ``gauge category'', because in the spectral
action approach to field theory inner fluctuations give rise to gauge fields. It has several
interesting properties, studied in \S\ref{sec:2.5}: a morphism between $D$ and $D'$ thus exists if and only if the
corresponding first order differential calculi coincide: $\Omega^1_D=\Omega^1_{D'}$;
$\mathsf{Mor}(D,D')$ is either empty or it contains exactly one element;
there never exists a final object in this category, but there may exist an initial object.

In \S\ref{sec:3}, we recall some basic ideas on the metric aspect of
noncommutative geometry, and stress the link with polarization
in geometric quantization.

In \S\ref{sec:4.1}, we introduce a Dirac operator for $M_n(\C)$ i.e.~a universal (initial) object in the category of inner fluctuations. In \S\ref{sec:4.2} we describe
a spectral triple for Moyal plane that is unitary equivalent to the
isospectral one of~\cite{GGISV04}, and prove in \S\ref{sec:4.3} that its ``polarization''
is a fluctuation of the spectral triple in \S\ref{sec:4.1}.

In \S\ref{sec:4.4} we discuss some metric properties of Moyal plane.
In \S\ref{sec:4.4.1} we give an alternative (shorter) proof of the formula in~\cite{MT11}
for the distance between translated states (including coherent
states).
In \S\ref{sec:4.4.2} we show that
eigenstates of the quantum harmonic oscillator,
with the spectral distance, form a metric space that is
convergent (for the Gromov-Hausdorff distance) to the Euclidean half-line when
$\theta\to 0$, where $\theta$ is the deformation parameter.

\section{A categorical approach to spectral triples}

We use the following notations: $\B(\HH)$ is the algebra
of all bounded linear operators on the Hilbert space $\HH$, $\mc{L}^2(\HH)$ the ideal
of Hilbert-Schmidt operators, $\mc{K}$ the $C^*$-algebra of compact
operators on a separable infinite-dimensional Hilbert space. When we
talk about \emph{states} of a pre $C^*$-algebra $\A$ we always mean
states of its $C^*$-closure. The set of
all states is denoted by $\mc{S}(\A)$.
The inner product and trace on a Hilbert space $\HH$ are denoted
$\inner{\,,\,}_{\HH}$ and $\tr_{\HH}$ respectively; the operator
norm is denoted $\|\,.\,\|_{\B(\HH)}$, and the subscript will be
omitted if there is no risk of ambiguity.
If $\E$ is a left $\A$-module, $\E'$ is a right $\A'$-module and
$\psi:\A\to\A'$ an isomorphims, we denote by $\E'\otimes_\psi\E$ the
quotient space of $\E'\otimes\E$ by the ideal generated by elements
$\xi'\psi(a)\otimes\xi-\xi'\otimes a\xi$, with $\xi\in\E$, $\xi'\in\E'$,
and $a\in\A$ (the tensor product over $\A\simeq\A'$).

\subsection{Noncommutative differential geometry}\label{sec:2.1}

Material in this section is mainly taken from~\cite{Con94,GVF01,Lan02,Mad95}.
In the spirit of~\cite{Con94}, the central notion for the description of noncommutative metric spaces is the notion of spectral triple.

A \emph{spectral triple} $(\A,\HH,D)$ is given by:
i) a separable complex Hilbert space $\HH$;
ii) a complex associative involutive algebra $\A$ with a bounded $*$-representation $\pi:\A\to\B(\HH)$;
iii) a self-adjoint operator $D$ on $\HH$ such that $[D,\pi(a)]$ is bounded and $\pi(a)(1+D^2)^{-1/2}$ is compact, for all $a\in\A$.
It is understood that $\pi(a)\cdot\mathrm{Dom}(D)\subset\mathrm{Dom}(D)$ for all $a\in\A$.

The spectral triple is \emph{unital} if $\A$ is a unital algebra and $\pi$ a unital
representation.
Without loosing generality, we will always assume that $\pi$ is faithful
and non-degenerate (one can always replace $\A$ by $\A/\ker\pi$ and $\HH$
by $\overline{\pi(\A)\HH}$), identify $\A$ with $\pi(\A)$ and omit the representation symbol $\pi$.

If $\A$ and $\HH$ are finite-dimensional, $D$ can be any self-adjoint
operator since all the conditions are trivially satisfied.

A spectral triple is \emph{even} if there is a \emph{grading} $\gamma$ on $\HH$,
i.e.~a bounded operator satisfying $\gamma=\gamma^*$, $\gamma^2=1$,
$[\gamma,a]=0\;\forall\;a\in\A$ and $\gamma D+D\gamma=0$.

\medskip

A first notion of equivalence between spectral triples
is unitary equivalence~\cite[\S7]{Var06}.
We say that two spectral triples $(\A,\HH,D)$ and $(\A',\HH',D')$
are \emph{unitary equivalent}, and write $(\A,\HH,D)\sim_{u}(\A',\HH',D')$,
if and only if there exists a unitary map $U:\HH\to\HH'$
such that: i) the map $\mathrm{Ad}_U:\B(\HH)\to\B(\HH')$ given by $\mathrm{Ad}_U(a):=UaU^*$
is an isomorphism between $\A$ and $\A'$, ii) $UD=D'U$.
If the spectral triples are even with gradings $\gamma$ and $\gamma'$, we add the further
requirement that $U\gamma=\gamma'U$.
Note that unitary equivalent spectral triples have associated isomorphic first order differential calculi:
$\mathrm{Ad}_U$ gives indeed a map $\Omega^1_D(\A)\to \Omega^1_{D'}(\A')$,
$a[D,b]\mapsto\mathrm{Ad}_U(a[D,b])=\mathrm{Ad}_U(a)[D',\mathrm{Ad}_U(b)]$.

A second relevant notion of equivalence is Morita equivalence. Before discussing it in the next section, let us recall
some basics definitions and fix some notations.
A right pre-Hilbert $\A$-module $\E$ for a pre $C^*$-algebra $\A$ is a right $\A$-module $\E$
endowed with an $\A$-valued Hermitian structure, that is a sesquilinear map $(\,,\,)_{\E}:\E\times\E\to\A$ satisfying $\,(\eta a,\xi b)_{\E}=a^*(\eta,\xi)_{\E}b\,$ and
$\,(\eta,\xi)_{\E}^*=(\xi,\eta)_{\E}\,$ for all $\eta,\xi\in\E $ and $a,b\in\A$,
plus $\,(\eta,\eta)_{\E}>0\,$ for all $\eta\neq 0$.

If $\A$ is a $C^*$-algebra and $\E$ is complete in the norm
$\|\eta\|_{\A}:=\sqrt{\|(\eta,\eta)_{\E}\|}$,
where $\|\,.\,\|$ is the $C^*$-norm of $\A$, then $\E$
is called a right Hilbert $\A$-module. Right Hilbert $\C$-modules
are simply complex Hilbert spaces. There is an analogous
definition for left modules (but in this case $(\,,\,)_{\E}$ 
is conjugate linear in the second argument).

From
an algebraic point of view, a smooth vector bundle $E\to M$ is described
by its set $\Gamma_0^\infty(E)$ of smooth sections vanishing at infinity.
This is a (right) pre-Hilbert $C^\infty_0(M)$-module, suitable for working with connections,
but not a Hilbert module since it is not complete. The
set $\Gamma_0(E)$ of all continuous sections vanishing at infinity
is a (right) Hilbert $C_0(M)$-module, and is the one
used to prove the strong Morita equivalence between $C_0(M)$ and
the algebra of right $C_0(M)$-linear endomorphisms of $\Gamma_0(E)$, cf.~e.g.~\cite[App.~A]{Lan02}.

\subsection{Morita equivalence}\label{sec:2.2}

Two rings $\A$ and $\B$ are Morita equivalent if there is an $\A$-$\B$-bimodule $\mc{M}$ and a $\B$-$\A$-bimodule $\mc{N}$ such that $\mc{M}\otimes_{\B}\mc{N}\simeq\A$ as $\A$-$\A$-bimodules and $\mc{N}\otimes_{\A}\mc{M}\simeq\B$ as $\B$-$\B$-bimodules (here $\otimes_{\A}$ is the tensor product over the algebra $\A$).
An elementary finite-dimensional example is $\A=M_n(\C),
\B=\C$, with $\mc{M}=\C^n$, resp.~$\mc{N}=\overline{\C^n}$, the space
of column vectors, resp.~row vectors, with $\A$-module
structure given by row-by-column multiplication: one has
$\C^n\otimes_{\C}\overline{\C^n}\simeq M_n(\C)$ and $\overline{\C^n}\otimes_{M_n(\C)}\C^n\simeq\C$.

Morita equivalence of rings is not suitable for $C^*$-algebras.
Take for example $\A=\K$, $\B=\C$, $\mc{M}:=\ell^2(\N)$ and $\mc{N}:=\ell^2(\N)^*$ the Hilbert space dual.
Then the unitary map{\footnote{$\{e_1,\ldots,e_n\}$ is the canonical
orthonormal basis of $\C^n$. $e_{ij}:=e_ie_j^*$ is the matrix
with $1$ in position $(i,j)$ and zero everywhere else.}} \mbox{$\mc{M}\otimes_{\C}\mc{N}\to\L^2(\mc{M})$},
defined by $e_i\otimes e_j^*\to e_{ij}$, has image that is only dense in $\K$
(the other way round, we could take $\A=\L^2(\mc{M})$ but this is not a $C^*$-algebra).
This motivates the notion of strong Morita equivalence~\cite{Rie74}.

\medskip

Let $\A, \B$ be $C^*$-algebra. A right Hilbert $\A$-module $\E$ is called \emph{full} if the linear span of $(\eta_1,\eta_2)_{\E}$, with $\eta_1,\eta_2\in\E$, is dense in $\A$ (in the $C^*$-algebra norm).
A $\B$-$\A$ bimodule
$_{\B}\E_{\A}$ is called a \emph{Hilbert bimodule} if it is
a right Hilbert $\A$-module, with $\A$-valued Hermitian structure
$(\,,\,)_{\E,\A}$, a left Hilbert $\B$-module, with $\B$-valued Hermitian structure
$(\,,\,)_{\E,\B}$, and the Hermitian structures are compatible in the following sense:
$$
\eta\,(\xi,\zeta)_{\E,\A}=(\eta,\xi)_{\E,\B}\,\zeta
\;,\qquad\forall\;\eta,\xi,\zeta\in {_{\B}\E_{\A}} \;.
$$
A $\B$-$\A$ Hilbert bimodule is \emph{full} if it is full both as right Hilbert $\A$-module
and left Hilbert $\B$-module.

$\A$ and $\B$ are \emph{strongly Morita equivalent}
if there exists a full $\B$-$\A$ Hilbert bimodule $_{\B}\E_{\A}$, called a \emph{Morita-equivalence bimodule} between $\A$ and $\B$, such that
$$
(\eta a,\eta a)_{\E,\B}\leq \|a\|^2(\eta,\eta)_{\E,\B} \;,\qquad
(b\eta,b\eta)_{\E,\A}\leq \|b\|^2(\eta,\eta)_{\E,\A} \;,
$$
for all $a\in\A$, $b\in\B$ and $\eta\in{_{\B}\E_{\A}}$.
We will use the notation $\B\rightleftharpoons\E\rightleftharpoons\A$.

\smallskip

If $\HH$ is a separable Hilbert space, one has $\K(\HH)\rightleftharpoons\HH\rightleftharpoons\C$.
The $\C$-valued Hermitian structure is simply the inner product $\inner{\,,\,}_{\HH}$ of $\HH$, and
$$
(\varphi,\psi)_{\E,\B}:=\varphi\inner{\psi,\,.\,}_{\HH} \;,\qquad\forall\;\varphi,\psi\in\HH.
$$

Strong Morita equivalence is an equivalence relation, see
e.g.~\cite[\S2.4]{Lan98book}:
\begin{list}{}{\itemsep=3pt \leftmargin=0em}
\item {\bf Reflexivity:} for any $C^*$-algebra $\A$, one has $\A\rightleftharpoons\A\rightleftharpoons\A$, where the left and right Hermitian structures
on the free module $\A$ are given by $(a,b)\mapsto ab^*$ and $(a,b)\mapsto a^*b$.

\item {\bf Symmetry:} if $\B\rightleftharpoons\E\rightleftharpoons\A$, then
$\A\rightleftharpoons\overline{\E}\rightleftharpoons\B$, where
$\overline\E$ is conjugate to $\E$. By definition,
elements of $\overline\E$ are in bijection with elements of $\E$, and we
will denote by $\eta^*\in\overline\E$ the element corresponding to $\eta\in\E$.
The left $\A$-$\B$-bimodule structure is defined as
$a\eta^*:=(\eta a^*)^*$ and $\eta^*b:=(b^*\eta)^*$.
The Hermitian structures are given by $(\eta^*,\xi^*)_{\overline\E,\A}:=
(\eta,\xi)_{\E,\A}$ and $(\eta^*,\xi^*)_{\overline\E,\B}:=(\eta,\xi)_{\E,\B}$.

\item {\bf Transitivity:} given two Morita-equivalence bimodules
${_{\mc{C}}\E_{\B}}$ and $_{\B}\E'\!\!{_{\A}}$, their tensor product
$_{\mc{C}}\E''\!\!{_{\A}}:={_{\mc{C}}\E_{\B}}\otimes_{\B}{_{\B}\E'\!\!{_{\A}}}$,
with Hermitian structures
$$
(\eta_1\otimes_{\B}\xi_1,\eta_2\otimes_{\B}\xi_2)_{\E'',\A}
:=\bigl(\xi_1,(\eta_1,\eta_2)_{\E,\B}\cdot\xi_2\bigr){_{\E',\A}}
$$
and
$$
(\eta_1\otimes_{\B}\xi_1,\eta_2\otimes_{\B}\xi_2)_{\E'',\mc{C}}
:=\bigl(\eta_1\cdot(\xi_1,\xi_2)_{\E',\B}\,,\eta_2\bigr){_{\E,\mc{C}}}
$$
is a Morita-equivalence bimodule between $\A$ and $\mc{C}$.
\end{list}

For strongly Morita equivalent $C^*$-algebras $\A$ and $\B$, if $\mc{M}$ is a Morita-equivalence bimodule between $\A$ and $\B$
and $\mc{N}$ is a Morita-equivalence bimodule between $\B$ and $\A$, one still has the bimodule isomorphisms $\mc{M}\otimes_{\B}\mc{N}\simeq\A$
and $\mc{N}\otimes_{\A}\mc{M}\simeq\B$, except that here the completed tensor product of Hilbert modules is used. This can be indeed taken as
an equivalent definition of strong Morita equivalence, see e.g.~\cite[Def.~4.9]{GVF01}. Another useful characterization is the following
\cite[Thm.~4.26]{GVF01}: $\A$ and $\B$ are strong Morita equivalent if{}f there is a full right Hilbert $\A$-module $\mc{E}$
such that $\B\simeq\mathrm{End}^0_{\A}(\E)$ (adjointable endomorphisms
of $\E$).

\subsection{Hermitian connections and fluctuations}\label{sec:2.3}

Any \emph{unital} spectral triple $(\A,\HH,D)$ has associated a canonical first order differential calculus $(\Omega^1_D(\A),\de_D)$ over $\A$, where $\Omega^1_D(\A)$ is the linear span of elements of the form $a[D,b]$, with $a,b\in\A$, and $\de_D a:=[D,a]$. In this review we will only consider differential calculi coming from spectral triples.

A connection on a right $\A$-module $\E$
is a linear map $\nabla:\E\to\E\otimes_{\A}\Omega^1_D(\A)$ satisfying
the Leibniz rule:
\begin{equation}\label{eq:Leibniz}
\nabla(\eta a)=(\nabla\eta)a+\eta\otimes_{\A}\de_D a \;,\qquad\forall\;a\in\A,\eta\in\E \;.
\end{equation}
Connections on left modules are defined in a similar way.

If $\E$ is a right pre-Hilbert $\A$-module, with Hermitian structure $(\,,\,)_{\E}$, the connection is called itself \emph{Hermitian} if:
\begin{equation}\label{eq:compat}
(\eta,\nabla\xi)_{\E}-(\nabla\eta,\xi)_{\E}=\de_{D}(\eta,\xi)_{\E} \;.
\end{equation}

Given a spectral triple $(\A,\HH,D)$,  an
Hermitian connection $\nabla_{\!D}$ on the right
pre-Hilbert $\A$-module $\E=\A$ (with Hermitian structure
$(a,b)_{\A}=a^*b)$ is
\begin{equation}\label{eq:nablaD}
\nabla_{\!D}(a)=1\otimes_{\A}[D,a]=1\otimes_{\A} \de_D a \;.
\end{equation}
This construction extends trivially to the module $\A^n$, for any
$n\geq 1$. Projecting the Hermitian structure and trivial connection,
one then put a
Hermitian connection, called the Grassmannian connection, on any
finitely generated projective right $\A$-module $\E$.

The vector space $\HHea=\E\otimes_{\A}\HH$ is a Hilbert space with inner product
$$
\inner{\eta_1\otimes\psi_1,\eta_2\otimes\psi_2}_{\HHea}:=
\inner{\psi_1,(\eta_1,\eta_2)_{\E}\,\psi_2}_{\HH} \;.
$$
It carries a natural bounded representation of
the algebra $\Ae:=\mathrm{End}_{\A}(\E)$ of right $\A$-linear endomorphisms of $\E$.
However $D$ is not $\A$-linear, so that the operator $1\otimes D$ on $\E\otimes\HH$ does not
define a (unbounded) operator on $\HHea$.
This can be cured defining the product $1\otimes_\nabla D$ by
\begin{equation}\label{eq:Dprime}
(1\otimes_\nabla D)(\eta\otimes_{\A}\psi):=\eta\otimes_{\A}D\psi+(\nabla\eta)\psi \;,\qquad\forall\;\eta\in\E,\psi\in\HH.
\end{equation}
For any Hermitian connection $\nabla$, eq.~\eqref{eq:Dprime} gives a well defined operator on $\HHea$, since thanks to \eqref{eq:Leibniz}, one has
$$
(1\otimes_\nabla D)(\eta a\otimes\psi)=(1\otimes_\nabla D)(\eta\otimes a\psi)
$$
for all $a\in\A$.{\footnote{We stress that $\nabla\eta\in\E\otimes_{\A}\Omega^1_D(\A)$ and
elements of $\Omega^1_D(\A)$ are represented by bounded operator on $\HH$: in the expression
$(\nabla\eta)\psi$ this action of $1$-forms on $\HH$ is understood, and the result is
then an element of $\E\otimes_{\A}\HH=\HHea$.}}
The result is a spectral triple $(\Ae,\HHea,1\otimes_{\nabla}D)$, cf.~\cite{CC06} or~\cite[\S10.8]{CM08} for the details.

\begin{df}
The triple $(\Ae,\HHea,1\otimes_{\nabla}D)$ is called \emph{fluctuation} of $(\A,\HH,D)$ in the direction
of $(\E,\nabla)$.
\end{df}

If $\E$ is a pre-Hilbert module, but is not finitely generated projective,
one can still perform the above construction (in this case $\E\otimes_{\A}\HH$
is a pre-Hilbert space, and must be completed),
with $\Ae$ replaced by any subalgebra of the algebra of adjointable endomorphisms of $\E$.
But now $(\Ae,\HHea,1\otimes_{\nabla}D)$ is not necessarily a spectral triple,
since the conditions involving the Dirac operator have to be verified case by case.
A set of conditions on $\E$ that guarantees that the above construction works
is given in~\cite{Mes09}.

\begin{lemma}\label{lemma:2.1}
Let $(\A,\HH,D)$ be a unital spectral triple and $(\E,\nabla)=(\A,\nabla_{\!D})$, with
canonical Hermitian structure $(a,b)_{\A}=a^*b$. Then
$$
(\Ae,\HHea,1\otimes_{\nabla}D)\sim_u(\A,\HH,D) \;.
$$
\end{lemma}
\begin{proof}
The left $\A$-module map $U:\E\otimes_{\A}\HH\to\HH$ defined by
$U(a\otimes_{\A}\psi):=a\psi:=\psi'$
is invertible with inverse $U^{-1}(\psi')=1\otimes_{\A}\psi'=a\otimes_{\A}\psi$.
It is isometric, hence unitary, since for all $a,a'\in\A$ and $\psi,\psi'\in\HH$:
\begin{align*}
\inner{a\otimes_{\A}\psi,a'\otimes_{\A}\psi'}_{\E\otimes_{\A}\HH} &=
\inner{\psi,(a,a')_{\E}\psi'}_{\HH}=
\inner{\psi,a^*a'\psi'}_{\HH} \\[2pt]
&=\inner{a\psi,a'\psi'}_{\HH}=
\inner{U(a\otimes_{\A}\psi),U(a'\otimes_{\A}\psi')}_{\HH} \;.
\end{align*}
Note that $1\otimes_{\nabla}D$ is self-adjoint on the domain
$\E\otimes_{\A}\mathrm{Dom}(D)=U^{-1}(\mathrm{Dom}(D))$,
and on this domain:
\begin{align*}
U(1\otimes_\nabla D)(a\otimes_{\A}\psi)
&= U(a\otimes_{\A}D\psi)+U\big((\nabla_{\!D}a)\psi\big) \\[2pt]
&= U(a\otimes_{\A}D\psi)+U\big(1\otimes_{\A}[D,a]\psi\big) \\[2pt]
&= aD\psi+[D,a]\psi=D(a\psi)=DU(a\otimes_{\A}\psi) \;.
\end{align*}
Hence $U(1\otimes_{\nabla}D)=DU$.
\end{proof}

The operator $\nabla_{\!D}$ makes sense only if $\A$ is unital.
If $\A$ is non-unital, one possibility is to choose a preferred unitization,
as explained in~\cite{GGISV04} (see also~\cite{Lod97} for the
corresponding problem in Hochschild homology). Since finite-dimensional pre $C^*$-algebras
are always unital, for the time being we ignore the complications of
non-unital algebras.

\subsection{Correspondences between spectral triples}\label{sec:mor}

In this section all algebras and Hilbert spaces are
finite dimensional. This has several
advantages: every finite-dimensional pre $C^*$-algebra is a
$C^*$-algebra, hence unital since by the classification
theorem $\A\simeq\bigoplus_{k=0}^mM_{n_k}(\C)$ for suitable $m,n_k$;
pre-Hilbert modules and Hilbert modules are the same (so that we can use the
same object to define connections and establish Morita equivalence).

\medskip

We want to define a category of spectral triples, called \emph{Morita category}
for its relation with Morita equivalence. Notice that from Lemma \ref{lemma:2.1} the pair
$(\E,\nabla)=(\A,\nabla_{\!D})$, that we would like to interpret as the identity morphism, transforms
$(\A,\HH,D)$ into a unitary equivalent spectral triple. This motivates the following definition.

\begin{df}
A \emph{correspondence}:
$$
(\A,\HH,D)\xrightarrow{(\E,\nabla,U)}(\A',\HH',D')
$$
between two finite-dimensional spectral triples is the datum of an $\A'$-$\A$ Morita equivalence
bimodule $\E$, a connection $\nabla:\E\to\E\otimes_{\A}\Omega^1_D(\A)$, and a unitary equivalence
$U:\E\otimes_{\A}\HH\to\HH'$ between the fluctuation $(\Ae,\HHea,1\otimes_{\nabla}D)$
and the target spectral triple $(\A',\HH',D')$.
\end{df}

Since here we assume `a priori' that the target of a correspondence is a spectral triple,
it is not necessary to explicitly require $\nabla$ to be Hermitian.
Moreover, since $\HH$ and $\HH'$ are finite-dimensional, $\E$ must be a finite-dimensional vector space as well.

\begin{df}
An \emph{inner fluctuation} is a correspondence of the form
$$
(\A,\HH,D)\xrightarrow{(\A,\nabla^\omega_{\!D},m)}(\A,\HH,D+\omega) \;,
$$
where 
$m:\A\otimes_{\A}\HH\to\HH$ is the multiplication map,
$\omega=\omega^*\in\Omega^1_D(\A)$ is the \emph{connection $1$-form} and
$$
\nabla^\omega_{\!D}(a):=1\otimes_{\A}\big\{[D,a]+\omega a\big\} \;.
$$
\end{df}

\begin{rem}
With our definition, correspondences are transformations of spectral triples that contain
fluctuations and unitary equivalences as subsets. A fluctuation $(\E,\nabla)$ is the same
as the correspondence $(\E,\nabla,\id_{\HH})$ (inner fluctuations strictly speaking are not
exactly fluctuations). A unitary equivalence $U$ is the same as the correspondence
$(\A,\nabla_{\!D},U\circ m)$.
\end{rem}

To interpret correspondences as morphisms in a category, we need to show that for each object there exists an identity morphism, that they can be composed, and that composition is associative.

\begin{prop}\label{prop:idmor}
For every finite-dimensional spectral triple there is a correspondence:
$$
(\A,\HH,D)\xrightarrow{(\A,\nabla_{\!D},m)}(\A,\HH,D)
$$
where $\nabla_{\!D}$ is given by \eqref{eq:nablaD}.
\end{prop}
\begin{proof}
It is a consequence of Lemma \ref{lemma:2.1}.
\end{proof}

To prove that the composition of correspondences is a correspondence, we first show that the natural composition of two fluctuations is still a fluctuation.

\begin{lemma}\label{lemma:2.4}
Given any two fluctuations
\begin{equation}\label{eq:flucone}
(\A,\HH,D)
\xrightarrow{(\E,\nabla)}
\bigl(\Ae,\HHea,1\otimes_{\nabla}D\bigr)
=:(\A',\HH',D') \;,
\end{equation}
and
\begin{equation}\label{eq:fluctwo}
(\A',\HH',D')
\xrightarrow{(\E',\nabla')}
\bigl(\A'_{\E'},\HH'_{\E'},1\otimes_{\nabla'}D'\bigr)
=:(\A'',\HH'',D'') \;,
\end{equation}
then:
\begin{itemize}\itemsep=0pt
\item[i)] there exists an $\A'$-$\A$ bimodule map
$$
\sigma\;:\;\Omega^1_{D'}(\A')\otimes_{\A'}\E\to \E\otimes_{\A}\Omega^1_{D}(\A)
$$
defined by
\begin{equation}\label{eq:sigma}
\sigma(a'[D',b']\otimes_{\A'}\eta)=a'\nabla(b'\eta)-a'b'\nabla(\eta) \;,
\end{equation}
for all $a',b'\in\A'$ and $\eta\in\E$;

\item[ii)] a connection $\nabla'':\E''\to\E''\otimes_{\A}\Omega^1_{D}(\A)$
on $\E'':=\E'\otimes_{\A'}\E$ is defined by:
\begin{equation}\label{eq:nabla2nd}
\nabla''(\eta'\otimes_{\A'}\eta)=(\id_{\E'}\otimes_{\A'}\sigma)\bigl\{\nabla'(\eta')\otimes_{\A'}\eta\bigr\}+\eta'\otimes_{\A'}\nabla(\eta) \;,
\end{equation}
for all $\eta'\in\E'$ and $\eta\in\E$;

\item[iii)] a fluctuation is given by $(\A,\HH,D)\xrightarrow{(\E'',\nabla'')}(\A'',\HH'',D'')$.
\end{itemize}
\end{lemma}

\begin{proof}
i) We can define an $\A'$-$\A$ bimodule map
$$
\sigma\;:\;\Omega^1_{D'}(\A')\otimes\E\to \E\otimes_{\A}\Omega^1_{D}(\A)
$$
by the same formula as in \eqref{eq:sigma}.
This is by construction a left $\A'$-module map,
but it is also a right $\A$-module map since by \eqref{eq:Leibniz} one has:
\begin{align*}
\sigma(a'[D',b']\otimes\eta a) &=a'\nabla(b'\eta a)-a'b'\nabla(\eta a) \\[2pt]
&=a'\nabla(b'\eta)a +a'b'\eta\otimes_{\A}d_D(a) -a'b'\nabla(\eta)a-
a'b'\eta\otimes_{\A}d_D(a) \\[2pt]
&=\sigma(a'[D',b']\otimes\eta)a \;,
\end{align*}
for all $a\in\A$.
Since
\begin{align*}
\sigma([D',a']b'\otimes\eta) &-\sigma([D',a']\otimes b'\eta) =
\sigma([D',a'b']\otimes\eta-a'[D',b']\otimes\eta)-\sigma([D',a']\otimes b'\eta) \\[2pt]
&=
\nabla(a'b'\eta)-a'b'\nabla(\eta)
-a'\nabla(b'\eta)+a'b'\nabla(\eta)
-\nabla(a'b'\eta)+a'\nabla(b'\eta) \\[2pt]
&=0
\end{align*}
for all $a',b'\in\A'$ and $\eta\in\E$, then $\sigma$ descends to an $\A'$-$\A$ bimodule map
(which we denote by the same symbol) with domain $\Omega^1_{D'}(\A')\otimes_{\A'}\E$.

\smallskip

\noindent ii)
With \eqref{eq:sigma} we can define a linear map $\nabla'':\E'\otimes\E\to\E''\otimes_{\A}\Omega^1_{D}(\A)$ as follows:
$$
\nabla''(\eta'\otimes\eta)=(\id_{\E'}\otimes_{\A'}\sigma)\bigl\{\nabla'(\eta')\otimes_{\A'}\eta\bigr\}+\eta'\otimes_{\A'}\nabla(\eta) \;.
$$
Thanks to the presence of $\sigma$, the image is in $\E''\otimes_{\A}\Omega^1_{D}(\A)=\E'\otimes_{\A'}\E\otimes_{\A}\Omega^1_{D}(\A)$ as requested.
Now, let $J\subset\E'\otimes\E$ be the ideal spanned by elements $\eta'a'\otimes\eta-\eta'\otimes a'\eta$,
with $a'\in\A'$, $\eta\in\E$ and $\eta'\in\E'$. We need to show that $J\subset\ker\nabla''$, so that
$\nabla''$ descends to a linear map (that we denote by the same symbol):
$$
\nabla'' : \E''\to\E''\otimes_{\A}\Omega^1_{D}(\A) \;.
$$
By the Leibniz rule
\begin{multline*}
\nabla''(\eta'a'\otimes\eta-\eta'\otimes a'\eta) =
(\id_{\E'}\otimes_{\A'}\sigma)\bigl\{\nabla'(\eta')a'\otimes_{\A'}\eta\bigr\}
-(\id_{\E'}\otimes_{\A'}\sigma)\bigl\{\nabla'(\eta')\otimes_{\A'} a'\eta\bigr\} \\[2pt]
+\eta'\otimes_{\A'}\sigma([D',a']\otimes_{\A'}\eta)+\eta'a'\otimes_{\A'}\nabla(\eta)-\eta'\otimes_{\A'}\nabla(a'\eta) \\[2pt]
=\eta'\otimes_{\A'}\sigma([D',a']\otimes_{\A'}\eta)+\eta'a'\otimes_{\A'}\nabla(\eta)-\eta'\otimes_{\A'}\nabla(a'\eta)
\end{multline*}
for all $a'\in\A'$. Using the definition of $\sigma$, we get
$$
\nabla''(\eta'a'\otimes\eta-\eta'\otimes a'\eta)=\eta'a'\otimes_{\A'}\nabla(\eta)-\eta'\otimes_{\A'} a'\nabla(\eta)=0 \;.
$$
Now we want to prove that $\nabla''$ is a connection. The Leibniz rule for $\nabla''$ follows
immediately from the Leibniz rule for $\nabla$ and the fact that $\sigma$ is a bimodule map:\vspace{-5pt}
\begin{multline*}
\nabla''(\eta'\otimes_{\A'}\eta a) =
(\id_{\E'}\otimes_{\A'}\sigma)\bigl\{\nabla'(\eta')\otimes_{\A'}\eta a\bigr\}+\eta'\otimes_{\A'}\nabla(\eta a) \\[2pt]
=(\id_{\E'}\otimes_{\A'}\sigma)\bigl\{\nabla'(\eta')\otimes_{\A'}\eta \bigr\}a+\eta'\otimes_{\A'}\nabla(\eta)a
+\eta'\otimes_{\A'}\eta\otimes_{\A}d_Da \\[2pt]
=\nabla''(\eta'\otimes_{\A'}\eta)a+(\eta'\otimes_{\A'}\eta)\otimes_{\A}d_Da \;.
\end{multline*}

\smallskip

\noindent iii) We want to show that
$(\A'',\HH'',D'')\sim_u\bigl(\A'',\E''\otimes_{\A}\HH,1\otimes_{\nabla''} D\bigr)$.
We only need to show that the isomorphism
$\HH''=\E'\otimes_{\A'}(\E\otimes_{\A}\HH)\to
\E''\otimes_{\A}\HH=(\E'\otimes_{\A'}\E)\otimes_{\A}\HH$
intertwines $1\otimes_{\nabla''}D$ with $D''$. Let us identify
these two spaces, and show that $1\otimes_{\nabla''}D=D''$ under
this identification.

Let $\Psi=\eta'\otimes_{\A'}\eta\otimes_{\A}\psi\in\HH''$, with
$\eta\in\E$, $\eta'\in\E'$ and $\psi\in\HH$.
Using \eqref{eq:Dprime}:
\begin{align*}
(1\otimes_{\nabla''}D)\Psi &=\eta'\otimes_{\A'}\eta\otimes_{\A}D\psi+\nabla''(\eta'\otimes_{\A'}\eta)\psi \\
        &=\eta'\otimes_{\A'}\eta\otimes_{\A}D\psi+
        \eta'\otimes_{\A'}\nabla(\eta)\psi+
        (\id_{\E'}\otimes_{\A'}\sigma)\bigl\{\nabla'(\eta')\otimes_{\A'}\eta\bigr\}\psi \\
D''\Psi &=\eta'\otimes_{\A'}D'(\eta\otimes_{\A}\psi)+(\nabla'\eta')(\eta\otimes_{\A}\psi) \\
        &=\eta'\otimes_{\A'}\eta\otimes_{\A}D\psi+\eta'\otimes_{\A'}\nabla(\eta)\psi+(\nabla'\eta')(\eta\otimes_{\A}\psi) \;.
\end{align*}
Hence
$$
(1\otimes_{\nabla''}D-D'')\Psi =
(\id_{\E'}\otimes_{\A'}\sigma)\bigl\{\nabla'(\eta')\otimes_{\A'}\eta\bigr\}\psi-(\nabla'\eta')(\eta\otimes_{\A}\psi) \;.
$$
By definition of connection and of $\Omega^1_{D'}(\A')$, $\nabla'(\eta')$ is a finite sum
$$
\nabla'(\eta')=\sum\nolimits_{ij}\eta'_i\otimes_{\A'}a'_{ij}[D',b'_{ij}]
$$
with $\eta'_i\in\E'$ and $a'_{ij},b'_{ij}\in\A'$. But for all $\eta\otimes_{\A}\psi\in\HH'$, using
\eqref{eq:Dprime} we find
$$
[D',b_{ij}](\eta\otimes_{\A}\psi)=\big\{\nabla(b_{ij}'\eta)-b_{ij}'\nabla(\eta)\big\}\otimes_{\A}\psi
=\sigma([D',b'_{ij}]\otimes_{\A'}\eta)\otimes_{\A}\psi
 \;,
$$
where the latter equality follows from the definition \eqref{eq:sigma}. Therefore
\begin{align*}
(\nabla'\eta')(\eta\otimes_{\A}\psi) &=\sum\nolimits_{ij}\eta'_i\otimes_{\A'}a'_{ij}[D',b_{ij}](\eta\otimes_{\A}\psi) \\
&=\sum\nolimits_{ij}\eta'_i\otimes_{\A'}a'_{ij}\sigma([D',b'_{ij}]\otimes_{\A'}\eta)\otimes_{\A}\psi \\
&=(\id_{\E'}\otimes_{\A'}\sigma)\left(\sum\nolimits_{ij}\eta'_i\otimes_{\A'}a'_{ij}[D',b'_{ij}]\otimes_{\A'}\eta\right)\psi \\
&=(\id_{\E'}\otimes_{\A'}\sigma)\bigl\{\nabla'(\eta')\otimes_{\A'}\eta\bigr\}\psi \;.
\end{align*}
This concludes the proof.
\end{proof}

Note that, strictly speaking, the map at point (iii) should be called a correspondence, since
a unitary equivalence $\E'\otimes_{\A'}(\E\otimes_{\A}\HH)\to (\E'\otimes_{\A'}\E)\otimes_{\A}\HH$
is understood. In the following, this kind of ``trivial'' isomorphisms (corresponding to associativity
of the tensor product operations) will always be omitted.

\smallskip

We denote by $\nabla'\odot\nabla$ the connection \eqref{eq:nabla2nd}.
Part (iii) of last lemma tells us that fluctuations can be composed, the composition
of $(\E,\nabla)$ in \eqref{eq:flucone} and $(\E',\nabla')$ in \eqref{eq:fluctwo} being $(\E'\otimes_{\A'}\E,\nabla'\odot\nabla)$.
We now show that correspondences can be composed as well. Before that, we need a 
``commutation rule'' between unitary equivalences and fluctuations.

\begin{prop}\label{prop:combo}
The composition of transformations:
\begin{equation}\label{eq:combo}
(\A,\HH,D)\xrightarrow{U}(\A',\HH',D')\xrightarrow{(\E',\nabla')}(\A'_{\E'},\HH'_{\E'},1\otimes_{\nabla'}D')
\end{equation}
is a correspondence. More precisely:
\begin{enumerate}
\item let $\E$ be the right $\A$-module given by $\E'$ as a vector space, with module structure 
$(\xi,a)\mapsto \xi\za a:=\xi\!\cdot\!\mathrm{Ad}_U(a)\;\forall\;\xi\in\E=\E'$ and $a\in\A$, and with $\A$-valued Hermitian structure
$(\xi,\eta)_{\E}:=\mathrm{Ad}_{U^*}(\xi,\eta)_{\E'}$ for all $\xi,\eta\in\E=\E'$;
\item let $\,\id\otimes U:\E\otimes_{\A}\HH\to\E'\otimes_{\A'}\HH'$ be the map defined by $\,\xi\otimes\psi\mapsto\xi\otimes U\psi$;
\item let $\nabla:\E\to\E\otimes_{\A}\Omega^1_D(\A)$ be the connection defined by $\nabla=(\id\otimes\mathrm{Ad}_{U^*})\nabla'$.
\end{enumerate}
Then, the correspondence $(\E,\nabla,\id\otimes U)$ has the same target as \eqref{eq:combo}:
\begin{equation}\label{eq:res}
(\A,\HH,D)\xrightarrow{(\E,\nabla,\id\otimes U)}(\A'_{\E'},\HH'_{\E'},1\otimes_{\nabla'}D') \;.
\end{equation}
\end{prop}

\begin{proof}
In this proof, let us call $U'=\id\otimes U:\E\otimes\HH\to\E'\otimes\HH'$ the map $\xi\otimes\psi\mapsto \xi\otimes U\psi$.
Then for any $a\in\A$:
$$
U'(\xi\za a\otimes \psi-\xi\otimes a\psi)=
\xi\mathrm{Ad}_{U}(a)\otimes U\psi-\xi\otimes Ua\psi=
\xi a'\otimes\psi'-\xi\otimes a'\psi' \;,
$$
where $a'=\mathrm{Ad}_U(a)\in\A'$ and $\psi'=U\psi\in\HH'$. Thus, $U'$ defines a map $\E\otimes_{\A}\HH\to\E'\otimes_{\A'}\HH'$
that we denote by the same symbol. Similarly $U'^*$ defines a map $\E'\otimes_{\A'}\HH'\to\E\otimes_{\A}\HH$
that we denote by the same symbol. Clearly $U'$ is unitary, indeed:
$$
\inner{\xi\otimes U\psi,\eta\otimes U\varphi}_{\E'\otimes_{\A'}\HH'}
=\inner{U\psi,(\xi,\eta)_{\E'}U\varphi}_{\HH'}=
\inner{\psi,U^*(\xi,\eta)_{\E'}U\varphi}_{\HH} \;.
$$
But $U^*(\xi,\eta)_{\E'}U=(\xi,\eta)_{\E}$, hence
$
\inner{\xi\otimes U\psi,\eta\otimes U\varphi}_{\E'\otimes_{\A'}\HH'}=
\inner{\xi\otimes\psi,\eta\otimes\varphi}_{\E\otimes_{\A}\HH} .
$
This proves that the target Hilbert space in \eqref{eq:res} is indeed $U'(\E\otimes_{\A}\HH)=\E'\otimes_{\A'}\HH'=\HH'_{\E'}$
(with equality, not just isomorphism). Furthermore, clearly $\A_{\E}:=\mathrm{End}_{\A}(\E)=\mathrm{End}_{\A'}(\E')=:\A'_{\E'}$,
so that in \eqref{eq:combo} and \eqref{eq:res} the target algebra is also the same (we stress again that we have
equality, not just isomorphism). It remains to show that $U'(1\otimes_{\nabla}D)=(1\otimes_{\nabla'}D')U'$.

Firstly, we show that $\nabla$ has image in $\E\otimes_{\A}\Omega^1_D(\A)$ and satisfies the Leibniz rule, so that the
connection is well defined. Since $\mathrm{Ad}_{U^*}$ maps $\Omega^1_{D'}(\A')$ to $\Omega^1_D(\A)$, then
$\id\otimes\mathrm{Ad}_{U^*}$ maps $\E'\otimes\Omega^1_{D'}(\A')$ to $\E\otimes\Omega^1_D(\A)$.
For all $\xi\in\E=\E'$, $\omega\in\Omega^1_{D'}(\A')$ and $a'=\mathrm{Ad}_{U}(a)\in\A'$,
$$
(\id\otimes\mathrm{Ad}_{U^*})(\xi a'\otimes\omega-\xi \otimes a'\omega)
=\xi\za a\otimes\mathrm{Ad}_{U^*}(\omega)-\xi \otimes a \mathrm{Ad}_{U^*}(\omega) \;.
$$
Thus, $\id\otimes\mathrm{Ad}_{U^*}$ gives a well-defined map $\E'\otimes_{\A'}\Omega^1_{D'}(\A')\to\E\otimes_{\A}\Omega^1_D(\A)$.
Furthermore,
\begin{multline*}
\nabla(\xi\za a)=(\id\otimes\mathrm{Ad}_{U^*})\nabla'(\xi a') \\[3pt]
=(\id\otimes\mathrm{Ad}_{U^*})\big\{(\nabla'\xi)a'+\xi\otimes_{\A'}[D',a']\big\}
=(\nabla\xi)a+\xi\otimes_{\A}[D,a] \;,
\end{multline*}
where as before $\xi\in\E=\E'$ and $a=\mathrm{Ad}_{U^*}(a')\in\A$. Hence the Leibniz rule
is satisfied.

Finally, for all $\xi\in\E$ and $\psi\in\HH$:
\begin{multline*}
U'(1\otimes_{\nabla}D)(\xi\otimes_{\A}\psi)=
U'(\xi\otimes_{\A}D\psi+(\nabla\xi)\psi) \\
=\xi\otimes_{\A}UD\psi+(\id\otimes\mathrm{Ad}_U)(\nabla\xi)\cdot U\psi
=\xi\otimes_{\A}D'U\psi+(\nabla'\xi)\cdot U\psi \\
=(1\otimes_{\nabla'}D')(\xi\otimes_{\A}U\psi)
=(1\otimes_{\nabla'}D')U'(\xi\otimes_{\A}\psi) \;,
\end{multline*}
hence $U'(1\otimes_{\nabla}D)=(1\otimes_{\nabla'}D')U'$.
\end{proof}

We are now ready to prove that the composition of two correspondences is again a correspondence.

\begin{dfprop}
The composition of two correspondences
$$
(\A,\HH,D)\xrightarrow{(\E,\nabla,U)}(\A',\HH',D')\xrightarrow{(\E',\nabla',U')}(\A'',\HH'',D'')
$$
is the correspondence
$$
(\A,\HH,D)\xrightarrow{(\E'',\nabla'',U'')}(\A'',\HH'',D'')
$$
given by
\begin{equation}\label{eq:correscomp}
\E'':=\E'\otimes_{\mathrm{Ad}_U}\E \;,\qquad
\nabla'':=\{(\id\otimes\mathrm{Ad}_{U^*})\nabla'\}\odot\nabla \;,\qquad
U'':=U'(\id\otimes U) \;.
\end{equation}
\end{dfprop}

\begin{proof}
Let's decompose each correspondence into a fluctuation plus a unitary equivalence:
\begin{multline*}
(\A,\HH,D)\xrightarrow{(\E,\nabla)}(\A_{\E},\HH_{\E},1\otimes_\nabla D)\xrightarrow{U}(\A',\HH',D')\xrightarrow{(\E',\nabla')}
\\[2pt]
\longrightarrow (\A'_{\E'},\HH'_{\E'},1\otimes_{\nabla'}D')\xrightarrow{U'}(\A'',\HH'',D'')
\end{multline*}
From Prop.~\ref{prop:combo}, denoting for obvious reasons still by $\E'$ the module at point 1 of the proposition,
it follows that the composition above produces the same result as
$$
(\A,\HH,D)\xrightarrow{(\E,\nabla)}\ldots
\xrightarrow{(\E',(\id\otimes\mathrm{Ad}_{U^*})\nabla')}\ldots\xrightarrow{\id\otimes U}\ldots
\xrightarrow{U'}(\A'',\HH'',D'')
$$
Since the composition of two fluctuations is a fluctuation, and the composition of two unitary
equivalences is a unitary equivalence, the composition above is a correspondence.
From Lemma \ref{lemma:2.4}, the relations \eqref{eq:correscomp} immediately follow.
\end{proof}

In order to have a category, morphisms must be defined as equivalence classes of correspondences.
We need then to specify when two correspondences are ``equivalent''.

\begin{df}\label{df:similarity}
Two correspondences
$$
(\A,\HH,D)\xrightarrow[(\E',\nabla',U')]{(\E,\nabla,U)}(\A',\HH',D')
$$
with the same source and target are \emph{similar} if there exists a unitary right $\A$-linear map
$V:\E\to\E'$ such that:
\begin{equation}\label{eq:similarity}
U'(V\otimes_{\A}\id_{\HH})=U \;.
\end{equation}
\end{df}

\begin{lemma}\label{lemma:similarity}
In the notations of previous definition, if $V$ is such a similarity, then the two connections
are related by $(V\otimes_{\A}\id_{\Omega^1_D(\A)})\nabla=\nabla'V$.
\end{lemma}

\begin{proof}
By hypothesis
\begin{align*}
1\otimes_{\nabla'}D
&=U'^*D'U'=(V\otimes_{\A}\id_{\HH})(U^*D'U)(V^*\otimes_{\A}\id_{\HH})
\\[2pt]
&=(V\otimes_{\A}\id_{\HH})(1\otimes_{\nabla}D)(V^*\otimes_{\A}\id_{\HH}) \;.
\end{align*}
From \eqref{eq:Dprime}, for all $\eta\in\E'$ and $\psi\in\HH$ one has
$(\nabla'\eta)\psi=(V\otimes_{\A}\id_{\HH})(\nabla V^*\eta)\psi$.
Since this is valid for all $\eta,\psi$, we get the thesis.
\end{proof}

\begin{prop}
Similarity of correspondences is an equivalence relation.
\end{prop}

\begin{proof}
\textit{Reflexivity}: $V=\id_{\E}$ is a self-similarity of $(\E,\nabla,U)$.
\textit{Symmetry}: if $V$ is a similarity between $(\E,\nabla,U)$ and $(\E',\nabla',U')$, then $V^*$
is a similarity between $(\E',\nabla',U')$ and $(\E,\nabla,U)$.
\textit{Transitivity}: if $V$ is a similarity between $(\E,\nabla,U)$ and $(\E',\nabla',U')$, and $V'$
is a similarity between $(\E',\nabla',U')$ and $(\E'',\nabla'',U'')$, then $V'V$ is a similarity between
$(\E,\nabla,U)$ and $(\E'',\nabla'',U'')$.
\end{proof}

\begin{rem}
Two fluctuations are similar if and only if they are equal.
Indeed, the unitary $V$ in Def.~\ref{df:similarity}, if it exists, is uniquely determined by $V\otimes_{\A}\id_{\HH}=U'^*U$.
But if the two correspondences are fluctuations, then $U=U'=\id_{\HH}$, and so $V$ is the identity as well.
\end{rem}

\begin{rem}
Two unitary equivalences $(\A,\nabla_D,U_1\circ m)$ and $(\A,\nabla_D,U_2\circ m)$ are similar if and only if 
$U_2^*U_1$ commutes with both $\A$ and $D$.
Indeed, the unitary $V:\A\to\A$ in Def.~\ref{df:similarity}, which now always exists, is given by $V(a)=\mathrm{Ad}_{U_2^*U_1}(a)$;
this is right $\A$-linear if and only if $U_2^*U_1$ commutes with all $a\in\A$, that is $V(a)=a$.
Since $D'U_i=U_iD$ for $i=1,2$, clearly $U_2^*U_1$ commutes with $D$.
\end{rem}

\noindent
An equivalence class of correspondences will be called \emph{M-morphism}\footnote{We borrow the terminology from \cite{BCL12}, although these are not the same morphisms.} (M for ``Morita'').

\begin{prop}
Equation \eqref{eq:correscomp} gives a well-defined composition rule for morphisms.
\end{prop}

\begin{proof}
Consider two correspondences
$$
(\A,\HH,D)\xrightarrow{(\E,\nabla,U)}(\A',\HH',D')\xrightarrow{(\E',\nabla',U')}(\A'',\HH'',D'')
$$
and suppose $V:\E\to\widetilde{\E}$ is a similarity between $(\E,\nabla,U)$ and $(\widetilde\E,\widetilde\nabla,\widetilde U)$,
and $V':\E'\to\widetilde{\E}'$ is a similarity between $(\E',\nabla',U')$ and $(\widetilde\E',\widetilde\nabla',\widetilde U')$.
We now show that the composition of the correspondences with and without tildas produces the same result, modulo similarity.

The composition of correspondences without tildas is the correspondence $(\E'',\nabla'',U'')$ given by \eqref{eq:correscomp},
and similar for the `tilda'-correspondences.

Since $V'$ is right $\A'$-linear, and $V$ intertwines the action of $\A_{\E}=\mathrm{End}_{\A}(\E)$ with the one of
$\A_{\widetilde\E}=\mathrm{End}_{\A}(\widetilde\E)$, there is a well-defined unitary right $\A$-module map
$$
V''=V'\otimes V:\E''=\E'\otimes_{\Ae}\E\to\widetilde\E''=\widetilde\E'\otimes_{\A_{\widetilde{\E}}}\widetilde\E \;.
$$
From $\widetilde U(V\otimes_{\A}\id_{\HH})=U$ and $\widetilde U'(V'\otimes_{\A'}\id_{\HH'})=U'$ we get
$\widetilde{U}''(V''\otimes_{\A}\id_{\HH})=U''$, thus proving \eqref{eq:similarity}.
\end{proof}

\begin{prop}
The correspondence in Prop.~\ref{prop:idmor} is a representative of the identity morphism of the object $(\A,\HH,D)$.
\end{prop}

\begin{proof}
We have to prove the commutativity of the following diagram (modulo similarity):

\medskip

 \begin{center}
 \begin{tikzpicture}[description/.style={fill=white,inner sep=4pt},normal line/.style={->},>=stealth]
 \matrix (m) [matrix of math nodes, row sep=4em, column sep=6em, text height=1.5ex, text depth=0.25ex]
 { (\A,\HH,D) & (\A',\HH',D') \\ (\A,\HH,D) & (\A',\HH',D') \\ };
 \path[normal line,font=\footnotesize] (m-1-1) edge node[above] {$ (\E,\nabla,U) $} node[below=2pt] { $1$ } (m-1-2);
 \path[normal line,font=\footnotesize] (m-1-2) edge node[right] {$ (\A',\nabla_{\!D'},m) $} node[left=1pt] { $2$ } (m-2-2);
 \path[normal line,font=\footnotesize] (m-1-1) edge node[left] {$ (\A,\nabla_{\!D},m) $} node[right=1pt] { $3$ } (m-2-1);
 \path[normal line,font=\footnotesize] (m-2-1) edge node[below] {$ (\E,\nabla,U) $} node[above=2pt] { $4$ } (m-2-2);
 \path[normal line,font=\footnotesize] (m-1-1) edge node[description] {$ (\E,\nabla,U) $} (m-2-2);
 \end{tikzpicture}
 \end{center}

\smallskip

\noindent
It is straightforward to check that a similarity between the composition of the arrows 1 and 2 and
the diagonal arrow is given by the multiplication map $m:\A'\otimes_{\A'}\E\to\E$,
while a similarity between the composition of the arrows
3 and 4 and the diagonal arrow is given by the multiplication map $m:\E\otimes_{\A}\A\to\E$.
\end{proof}

\begin{prop}\label{prop:ass}
The composition of M-morphisms is associative.
\end{prop}

\begin{proof}
We will show that the composition of correspondences is associative, from which Prop.~\ref{prop:ass} follows.
Consider three correspondences:
$$
(\A,\HH,D) \xrightarrow{(\E,\nabla,U)} (\A',\HH',D') \xrightarrow{(\E',\nabla',U')} (\A'',\HH'',D'')
 \xrightarrow{(\E'',\nabla'',U'')} (\A''',\HH''',D''')
$$
The tensor product of bimodules is associative: there is a canonical isomorphism between
$\E''\otimes_{\A''}(\E'\otimes_{\A'}\E)$ and $(\E''\otimes_{\A''}\E')\otimes_{\A'}\E$ that we omit,
and we simply denote by $\E''':=\E''\otimes_{\A''}\E'\otimes_{\A'}\E$ the resulting bimodule.
The composition of unitary operators is associative too, the resulting unitary being
$U''':=U''\big(\id\otimes U'(\id\otimes U)\big):\E''\otimes_{\A''}\E'\otimes_{\A'}\E\otimes_{\A}\HH\to\HH'''$.
We now show that the also composition of connections is associative. This is in fact a simple consequence
of the fact that the connection is fixed by the rest of the data (source and target spectral triple, bimodule and
unitary).

Let $\nabla_U'=(\id\otimes\mathrm{Ad}_{U^*})\nabla'$ and 
$\nabla_{U'}''=(\id\otimes\mathrm{Ad}_{U'^*})\nabla''$.
Let
$\nabla_L:=\nabla''_{U'}\odot (\nabla'_U\odot\nabla)$
and
$\nabla_R:=(\nabla''_{U'}\odot\nabla'_U)\odot\nabla$.
We have two morphisms
$$
(\A,\HH,D)\xrightarrow[(\E''',\nabla_R,U''')]{(\E''',\nabla_L,U''')}(\A''',\HH''',D''')
$$
with the same source and target. From $1\otimes_{\nabla_L}D=U'''^*D'''U'''=
1\otimes_{\nabla_R}D$ and \eqref{eq:Dprime} one easily deduces $\nabla_L=\nabla_R$.
\end{proof}

\begin{df}\label{df:Morcat}
We call ``Morita category'' the category whose objects are spectral triples, and whose
morphisms are M-morphisms.
\end{df}

\begin{ex}\label{rem:2.8}
Not every M-morphism is an isomorphism. Consider the spectral triple
$(M_2(\C),\C^2,D)$, with
$$
D:=\bigg(\!\begin{array}{cc}
0 & 1 \\ 1 & 0
\end{array}\!\bigg) \;.
$$
Since $\Omega^1_D(M_2(\C))=M_2(\C)$
(one easily finds $a,b\in M_2(\C)$ such that $a[D,b]$ is the identity matrix), then 
$D'=D+\omega$ is an inner fluctuation of $D$ for any $\omega=\omega^*\in M_2(\C)$.
So in particular, for $\omega=-D$ one gets a fluctuation
$(M_2(\C),\C^2,D)\to (M_2(\C),\C^2,0)$. But this is not an isomorphism,
since $\Omega^1_0(M_2(\C))=\{0\}$ and no non-zero Dirac operator can be
obtained as an inner fluctuation of $0$.
\end{ex}

From the category in Def.~\ref{df:Morcat} one can obtain two subcategories, with the same objects but
less morphisms. A first subcategory is the one with morphisms given by unitary equivalences only:
every morphism is an isomorphism, hence it is a groupoid. Note that the converse is
not true: not every isomorphism of the category \ref{df:Morcat} is represented by a unitary
equivalence. For example in \S\ref{sec:4.3} we will construct spectral triples $(M_n(\C),\C^n\otimes\C^2,D_n)$
that, for $n\geq 2$, are all isomorphic in the category \ref{df:Morcat}, but they are not unitary equivalent
since the Dirac operators have different spectra.

A second subcategory is the one with morphisms given by inner fluctuations only
(and these are not always invertible).  It is studied in the next section,
where we show that the composition of two inner fluctuations is indeed an
inner fluctuation. Since every inner fluctuation has the same pair $(\A,\HH)$ as
source and target (only the Dirac operator changes), we can fix the algebra and
the Hilbert space. Moreover, in the spectral action approach to field theory, inner
fluctuations give rise to gauge fields. We thus refer to this category
as the \emph{gauge category} of $(\A,\HH)$.

\subsection{Inner fluctuations}\label{sec:2.5}

Fix finite-dimensional $\A$ and $\HH$. It is well known that inner fluctuations form a semigroup
extending the group of unitary elements of $\A$. This has been proved also for
real spectral triples, even when the first order condition is not satisfied \cite{CCvS13}.
We reproduce here the argument of \cite[Prop.~5(ii)]{CCvS13}, simply dropping the real structure, which plays no role in this paper.

We define a category $\mf{C}(\A,\HH)$ whose objects are self-adjoint operators $D\in\B(\HH)$, and $\mathsf{Mor}(D,D')$ is the set of $A=A^*\in\Omega^1_D(\A)$ such that $D'=D+A$. Notice that, if $A=\sum_ia^i[D,b_i]$, with $a^i,b_i\in\A$, then
\begin{multline*}
[D',c]=[D,c]+\sum\nolimits_i[a^i[D,b_i],c]
\\[4pt]
=[D,c]+\sum\nolimits_i\Big\{
[a^i,c][D,b_i]+a^i[D,b_ic]
-a^ib_i[D,c]-a^ic[D,b_i]
\big\}
\end{multline*}
for any $c\in\A$. Hence $\Omega^1_{D'}(\A)\subset\Omega^1_D(\A)$.
Therefore, if $A\in\mathsf{Mor}(D,D')$
and $A'\in\mathsf{Mor}(D',D'')$, clearly
$$
D''=D'+A'=D+A+A'
$$
and since $A+A'\in \Omega^1_D(\A)+\Omega^1_{D'}(\A)=\Omega^1_D(\A)$,
then $A+A'\in\mathsf{Mor}(D,D'')$, so that the composition of inner fluctuations is
still an inner fluctuation. It is associative since the sum of operators is associative,
and $0\in\mathsf{Mor}(D,D)$ is the identity morphism of the object $D$.
Note that, for any $D$ and $D'$,
$\mathsf{Mor}(D,D')$ is either the empty set or the set with
only one element $A=D'-D$, depending on whether $D'-D$ belongs
or not to $\Omega^1_D(\A)$.

\begin{rem}
Since $\Omega^1_{D'}(\A)\subset\Omega^1_D(\A)$, inner fluctuations cannot
increase the size of the differential calculus. Example \ref{rem:2.8}
provides an example where the latter inclusion is proper, as in that
case $\Omega^1_{D'}(\A)=\{0\}$ and $\Omega^1_{D}(\A)=M_2(\C)$.
\end{rem}

\begin{prop}
$A\in\mathsf{Mor}(D,D')$ is an isomorphism if and only if $\Omega^1_{D'}(\A)=\Omega^1_D(\A)$;
in this case, the inverse is $A':=-A$.
\end{prop}

\begin{proof}
If $\Omega^1_{D'}(\A)=\Omega^1_D(\A)$, then $-A$ is an element of $\Omega^1_{D'}(\A)$ and
it is inverse to $A$. Viceversa, by the considerations above, if there exists an inner fluctuation
$A'\in\mathsf{Mor}(D',D)$, we have both $\Omega^1_{D'}(\A)\subset\Omega^1_D(\A)$
and $\Omega^1_{D'}(\A)\supset\Omega^1_D(\A)$, proving that these two sets coincide.
Moreover, if $A'$ is inverse to $A$, clearly $A+A'$ must be zero, hence $A'=-A$.
\end{proof}

\begin{rem}\label{prop:2.11}
The operator $D$ in Ex.~\ref{rem:2.8} is an initial object in the
category $\mf{C}(M_2(\C),\C^2)$.
Indeed, since $\Omega^1_D(M_2(\C))=M_2(\C)$, for any
$D'=D'^*\in M_2(\C)$ we can find an inner fluctuation $A=D'-D\in\mathsf{Mor}(D,D')$.
\end{rem}

The initial object does not always exists. Consider the example
of $\A=\C$ and $\HH=\C$. Dirac operators are real numbers, so that
the set of object of $\mf{C}(\C,\C)$ is $\R$, and they all commute
with the algebra $\C$. Therefore $\Omega^1_D(\C)=\{0\}$ for any
$D$, and the only inner fluctuation is the trivial one.
This category has no initial object, and no final object.

\smallskip

While in some examples the initial object may exists, the
final object never exists.

\begin{prop}
Let $\HH\neq\{0\}$ be any non-zero finite-dimensional Hilbert space and $\A\subset\B(\HH)$
a $C^*$-algebra. Then, $\mf{C}(\A,\HH)$ has no final object.
\end{prop}

\begin{proof}
Since the identity operator $\id_{\HH}$ commutes with $\A$,
$\Omega^1_{\id_{\HH}}(\A)=\{0\}$ and the only fluctuation of $\id_{\HH}$ is $\id_{\HH}$ itself. Hence $\mathsf{Mor}(\id_{\HH},0)=\emptyset$ and $0$ is not a final object.
On the other hand, no non-zero $D$ can be a final object, since in this case $\mathsf{Mor}(0,D)=\emptyset$.
\end{proof}

In particular, the example in last proof shows that not all finite-dimensional
Dirac operators can be fluctuated to zero.

\smallskip

$\mf{C}(\C,\C)$ is the category with set of objects
equal to $\R$, $\mathsf{Mor}(x,y)=\emptyset$ if $x\neq y$
and $\mathsf{Mor}(x,y)=\{0\}$ if $x=y$.
We can give a more concrete description of $\mf{C}(M_n(\C),\C^n)$ as well.
Let us call $D$ \emph{trivial} if proportional to the identity.

\begin{prop}
$\Omega^1_D(M_n(\C))=M_n(\C)$ for any non-trivial object $D$ in $\mf{C}(M_n(\C),\C^n)$.
\end{prop}

\begin{proof}
Let $e_{ij}\in M_n(\C)$ be the matrix with $1$ in position $(i,j)$ and zero everywhere else.
If $D$ is non-trivial, $\Omega^1_D(M_n(\C))$ has at least one non-zero element $\omega$. Suppose
the matrix element $\omega_{kl}$ is not zero. Then
$(\omega_{kl})^{-1}e_{ik}\cdot\omega\cdot e_{lj}=e_{ij}$ belongs to $\Omega^1_D(M_n(\C))$ for
all $i,j$, thus concluding the proof.
\end{proof}

In the latter example, $\Omega^1_D(M_n(\C))$ is either $\{0\}$, if $D$ is trivial, or the whole
$M_n(\C)$. As in Remark \ref{prop:2.11}, any non-trivial $D$ is a initial object in the category,
and any two non-trivial Dirac operators are isomorphic.

\section{Matrix geometries emergent from a point}

Here we discuss how a matrix geometry emerges from a non-trivial
differentiable structure on the space with one point.
One finds a spectral triple based on the algebra $\mc{S}(\N^2)$ of
rapid decay matrices, which is equivalent to the isospectral spectral
triple of Moyal plane, by means of the matrix basis.

\subsection{Spectral distance and polarization}\label{sec:3}

We begin with a general construction, which allows us to treat any normal
state as a vector state, and is compatible with the metric aspect of
noncommutative geometry. 

Given a (not necessary unital, nor finite-dimensional)
spectral triple $(\A, \HH, D)$, the set of states $\mc{S}(\A)$ is an extended metric
space\footnote{By extended metric space we mean a pair $(X,d)$ with $X$ a set
and $d:X\times X\to[0,\infty]$ a symmetric map satisfying the
triangle inequality and such that $d(x,y)=0$ if{}f $x=y$.
The only difference with an ordinary metric space is that the
value $+\infty$ for the distance is allowed.} with distance 
$$
d_D(\varphi,\varphi')=\sup_{a=a^*\in\A}\big\{\varphi(a)-\varphi'(a)\,:\,
\|[D,a]\|_{\B(\HH)}\leq 1 \big\}
$$
for all $\varphi,\varphi'\in\mc{S}(\A)$. This is usually called \emph{Connes metric} or \emph{spectral distance}.

\smallskip

A state $\varphi$ is
normal if it admits a (non-necessarily unique) density matrix $\rho$, that is a positive operator with
trace $1$ such that
$$
\varphi(a):=\tr_{\HH}(\rho\hspace{1pt}a) \;,\qquad a\in\A.
$$
A vector state is a normal state whose density matrix $\rho$ has rank
$1$, that is $\rho=\psi\psi^\dag$ for some vector $\psi\in\HH$. Then
\begin{equation}
\varphi(a)=\inner{\psi,a\psi}_{\HH} \;,\qquad a\in\A. \label{vectorstate}
\end{equation}

Any normal states is a vector state in the representation of $\A$ on the
Hilbert space $\mathcal{L}^2(\HH)$ of Hilbert-Schmidt operators, with inner product:
$$
\inner{A,B}_{\mc{L}^2(\HH)}=\tr_{\HH}(A^*B) \;.
$$
Indeed $\mathcal{L}^2(\HH)$ is a two-sided
ideal in $\B(\HH)$, so that $\A\subset\B(\HH)$ has a natural representation on $\mathcal{L}^2(\HH)$
given by the composition of operators.  Moreover, since $\mathcal{L}^2(\HH)$ is
isomorphic to $\HH\otimes\HH^*$, for any density matrix
$\rho\in\B(\HH)$ there exists an Hilbert-Schmidt operator $\eta$ such
that $\rho=\eta\eta^*$. By cyclicity of the trace
$$
\tr_{\HH}(\rho\,.\,)=\inner{\eta,\,.\,\eta}_{\mc{L}^2(\HH)}.
$$

Thus by replacing the original spectral triple $(\A,\HH,D)$ with a new
one with Hilbert space $\mathcal{L}^2(\HH)$, any
 normal state is a vector state. To guarantee that the new 
 triple is  metrically equivalent to the initial one, a possibility is to take $(\A,\mathcal{L}^2(\HH),\mathcal{D})$
where $\mathcal{D}$ is the operator
$$
\mathcal{D}(A):=[D,A] \;.
$$
Under the identification $\mathcal{L}^2(\HH)\simeq\HH\otimes\HH^*$
the operator $\mathcal{D}$ is self-adjoint on the domain
$\mathrm{Dom}(D)\otimes\mathrm{Dom}(D)^*$ and  we have
\begin{equation}\label{eq:calDD}
[\mathcal{D},a]=[D,a]\otimes\id_{\HH^*}
\end{equation}
for any $a\in\A$. Hence $\mathcal{D}$ has bounded commutators with $\A$.
The resolvent condition is not necessarily satisfied
(it must be checked case by case), so that $(\A,\mathcal{L}^2(\HH),\mathcal{D})$
is not necessarily a spectral triple (it is if $\HH$ is
finite-dimensional). However the spectral distance is still well
defined and as an immediate consequence of \eqref{eq:calDD} one has for any $\varphi,\varphi'\in \mc{S}(\A)$
$$
d_D(\varphi,\varphi')=d_{\mc{D}}(\varphi,\varphi').
$$

\smallskip

If the original spectral triple is even, with $\HH=\HH_0\otimes\C^2$, standard grading and
$$
D=\mat{0 & D_+ \\ D_- & 0} \;,
$$
a similar construction yields a triple $(\A,\mc{L}^2(\HH_0)\otimes\C^2,\mc{D})$,
where now
$$
\mc{D}=\mat{0 & \mc{D}_+ \\ \mc{D}_- & 0} \;,
$$
and $\mc{D}_\pm(A)=[D_\pm,A]$. If the resolvent condition is
satisfied, with the canonical grading we get a new even spectral triple.

\begin{lemma}\label{lemma:1}
For both $(\A,\HH_0\otimes\C^2,D)$ and $(\A,\mathcal{L}^2(\HH_0)\otimes\C^2,\mathcal{D})$, the
spectral distance between $\varphi$ and $\varphi'\in \mc{S}(\A)$ is the supremum of
$|\varphi(a)-\varphi'(a)|$ over all $a\in\A$ satisfying
\begin{equation}\label{eq:SDpiu}
\|[D_+,a]\|_{\B(\HH_0)}\leq 1\quad\mathrm{and}\quad\|[D_-,a]\|_{\B(\HH_0)}\leq 1  \;.
\end{equation}
Equivalently, one can take the supremum over self-adjoint elements, with only one condition
$$
\|[D_+,a]\|_{\B(\HH_0)}\leq 1 \;.
$$
\end{lemma}
\begin{proof}
The proof is a simple observation. Since $D_-=D_+^*$ one has
$$
[D,a]^*[D,a]=
\mat{[D_-,a]^*[D_+,a] & 0 \\ 0 & [D_+,a]^*[D_+,a]} \;.
$$
It follows that $\|[D,a]\|_{\B(\HH)}$ is the maximum between $\|[D_-,a]\|_{\B(\HH_0)}$ and $\|[D_+,a]\|_{\B(\HH_0)}$.
If $a=a^*$, then $[D_+,a]$ and $[D_-,a]=-[D_+,a]^*$ have the same norm.
\end{proof}

We will see in Moyal example that if $T_S= (\A,\HH_0\otimes\C^2,D)$ is the spectral triple constructed with
the irreducible (Schr{\"o}dinger) representation, then $(\A,\mathcal{L}^2(\HH_0)\otimes\C^2,\mathcal{D})$
is unitary equivalent to the isospectral spectral triple $T_W$ constructed with the GNS (Wigner) representation
associated to the trace. The passage from $T_S$ to $T_W$ is the opposite of polarization in geometric
quantization (see e.g.~\cite{Wei97}).

\medskip

Let us conclude with some basic definition about dimension and integrals (see e.g.~\cite{Con94,GVF01,Lan02}).
If there exists $t\in\R$ such that $\pi(a)(1+D^2)^{-t/2}\in\mathcal{L}^{(1,\infty)}(\HH)$ is in the Dixmier ideal
for all $a\in\A$, we say that the spectral triple is finite-dimensional; we  call \emph{metric dimension} the inf
of such $t$'s; if $\A$ is unital, a necessary and sufficient condition is that $(1+D^2)^{-t/2}\in\mathcal{L}^{(1,\infty)}(\HH)$;
in both cases, one defines a cyclic integral by
$$
\nint a:=\tr_\omega(\pi(a)(1+D^2)^{-t/2})
$$
where $\tr_\omega$ is the Dixmier's trace. If $D$ is invertible, $(1+D^2)^{-t/2}$ can be replaced by $|D|^{-t}$.

In the canonical example of a non-compact manifold $M$, the metric dimension coincides with the dimension
of $M$ if we use the algebra $\mc{S}(M)$ of Schwartz functions, which unlike $C^\infty_0(M)$ is made of integrable functions.

Assume that the noncommutative space has finite dimension.
If $\A$ is non-unital, we say that the space has a \emph{finite volume} $V>0$ when
$$
\nint 1=V<\infty \;.
$$
Note that $\tint a$ is finite for all $a\in\A$, thus if $\A$ is unital previous
condition is always satisfied.

%%% ======================================================================

\subsection{On the geometry of the $1$-point space}\label{sec:4.1}

To put a non-trivial geometrical structure on a single point one can proceed as follows.
Consider a space with $n$ indistinguishable points. More precisely, 
let us consider the set $I_n=\{1,\ldots,n\}$ with equivalence relation $j\sim k$ for all $j,k\in I_n$,
so that the quotient space is the space with one point.
The point of view pioneered in \cite{Con94} is that geometric informations about the equivalence
relation are captured by the groupoid algebra $\C\mc{G}$ of the graph $\mc{G}$ of the equivalence
relation. In the above example, $\mc{G}$ is the groupoid of pairs of elements of $I_n$, and
$\C\mc{G}$ is the crossed product $C^*$-algebra $C(I_n)\rtimes \Z_n$, which is isomorphic to $M_n(\C)$.\footnote{More generally,
$C(G)\rtimes G\simeq M_n(\C)$ for any group $G$ with $n$ elements acting on itself by left multiplication \cite[Lemma 2.50]{Wil07}.}

This point of view allows to put non-trivial differentiable structures on the space with one point. If $n=1$ we
are forced to use a degenerate representation, but for $n\geq 2$ we can define non-trivial differentiable structures
(i.e.~spectral triples whose associated differential calculus is not identically zero) using the standard representation
of the algebra. Requiring the spectral distance to be finite, so that we have a compact quantum metric space, one is
forced to double the Hilbert space. We can then consider a even spectral triple of the form
$(\A_n,\HH_n,D_n,\gamma_n)$, where $\A_n=M_n(\C)$ is represented on $\HH_n=\C^n\otimes\C^2$ by row-by-column
multiplication on the first factor, and let
\begin{equation}\label{eq:Dn}
D_n=\sqrt{\frac{2}{\theta}}\mat{0 & X_n^* \\ X_n & 0} \;,\qquad \gamma_n=\mat{1 & \;\,0 \\ 0 & \!-1} \;,
\end{equation}
with $X_n\in M_n(\C)$. A possible choice is
$$
X_n=\begin{pmatrix}
0 & 0 & 0 & \ldots&0 \\
1 & 0 & 0 & \ldots &0\\
0 &\sqrt{2} & 0 & \ldots&0 \\
\vdots & \ddots & \vdots & \vdots &\vdots\\
0 & \ldots& 0 & \sqrt{n-1}& 0
\end{pmatrix} \;,
$$
with $\theta>0$ a parameter, leading to the spectral triple studied in \cite[Sec.~4.2]{CDMW09}.\footnote{The parameter $\theta$,
and the factor $\sqrt{2/\theta}$ in \eqref{eq:Dn}, are introduced to have the same normalization used in~\cite[\S4]{CDMW09} (although there
the factor $1/\sqrt{\theta}$ is included in the definition of $X_n$).}

Let $e_{ij}\in M_n(\C)$ be the matrix with $1$ in position $(i,j)$ and zero everywhere else.
The set of $1$-forms $\Omega^1_{D_n}(\A_n)$ is the rank $2$
free $M_n(\C)$-module with basis elements the two matrices
\begin{equation}\label{eq:omegauno}
\begin{split}
\sqrt{\frac{\theta}{2}}\sum_{k=0}^{n-1}\,e_{k0}\otimes 1_2\,[D_n,e_{1k}\otimes 1_2]=
\bigg(\!\begin{array}{ll}0 & 1_n\!\! \\ 0 & 0\end{array}\!\bigg)
 \;,\\
\sqrt{\frac{\theta}{2}}\sum_{k=0}^{n-1}\,e_{k1}\otimes 1_2\,[D_n,e_{0k}\otimes 1_2]=
\bigg(\!\begin{array}{ll}0 & 0 \\ 1_n\! & 0\end{array}\!\bigg)
 \;,
\end{split}
\end{equation}
i.e.~the set of matrices {\footnotesize$\bigg(\!\begin{array}{cc}0 & a \\ b & 0\end{array}\!\bigg)
$}, with $a,b\in M_n(\C)$.
This in particular means that any self-adjont operator anticommuting with $\gamma$ can be obtained as an inner fluctuation of $D_n$.

In the language of \S\ref{sec:2.5}, $D_n$ is a universal (initial) object in the category
of even Dirac operators on $\HH_n$.

\subsection{Moyal plane as a matrix geometry}\label{sec:4.2}

The spectral triple $(\mc{S}(\R^2),L^2(\R^2)\otimes\C^2,\D)$, with $\D$
the Dirac operator of $\R^2$, can be quantized replacing the pointwise product with Moyal star product.
This yields the isospectral spectral triple of Moyal plane studied
in~\cite{GGISV04}, here and in the following denoted $T_W$,
which is unitary equivalent to the spectral triple
$(\mc{S}(\N^2),\mathcal{L}^2(\ell^2(\N))\otimes\C^2,\mc{D})$,
where
$\mc{S}(\N^2)$ is the Fr\'echet pre $C^*$-algebra of rapid decay matrices
(with natural seminorms, recalled for example in~\cite{CDMW09}),
$\mathcal{L}^2(\ell^2(\N))$ is the space of Hilbert-Schmidth operators
on $\ell^2(\N)$,
$$
\mc{D}=\sqrt{\frac{2}{\theta}}\mat{0 & \mc{D}_+ \\ \mc{D}_- & 0} \;,
$$
and
$$
\mc{D}_-(A):=[\mf{a},A] \;,\qquad
\mc{D}_+(A):=[\mf{a}^\dag,A] \;.
$$
Here $\mf{a}^\dag$ and $\mf{a}$ denote the creation and
annihilation operators:
$$
\mf{a}^\dag\ket{n}=\sqrt{n+1}\ket{n+1} \;,\qquad
\mf{a}\ket{n}=\sqrt{n}\ket{n-1} \;,
$$ 
with $\ket{n}$, $n\geq 0$, the canonical orthonormal basis of
$\ell^2(\N)$.
This spectral triple has metric dimension $2$ \cite{GGISV04} and infinite
volume, as $\,\tint 1=\int_{\R^2}1\,\de^2x =\infty$.

A polarization is given by the spectral
triple $T_S = (\mc{S}(\N^2),\ell^2(\N)\otimes\C^2,D)$, where
\begin{equation}\label{eq:Dirred}
D=\sqrt{\frac{2}{\theta}}\,\bigg(\!\begin{array}{cc}
0\, & \mathfrak{a} \\ \mathfrak{a}^\dag\! & 0
\end{array}\!\bigg) \;.
\end{equation}
The domain of $D$ is given by vectors ${^t}(v_1,v_2)$, with $v_i\in\ell^2(\N)$
such that $\mf{a}\hspace{1pt}v_i$ and $\mf{a}^\dag\hspace{1pt}v_i$ are
square summable, for any $i=1,2$.
The spectrum of $D^2$ is $\{\frac{2}{\theta}n\}_{n\in\N}$.
On the orthogonal complement of $\ket{0}\otimes\binom{0}{1}$ (the kernel of $D$),
$|D|^{-s}$ is trace-class for any $s\geq 2$ and
$$
\nint 1=\mathrm{Res}_{z=2}\tr(|D|^{-z})
=\theta\,\mathrm{Res}_{z=2}\zeta(z/2)
=2\theta\,\mathrm{Res}_{z=1}\zeta(z)
=2\theta \;,
$$
where $\zeta$ is the Riemann's zeta-function (that has a simple pole
at $z=1$ with residue equal to $1$). So $T_S$ has metric dimension $2$
and finite volume, although it is metrically equivalent to $T_W$. Also
$T_S$ is not defined for $\theta=0$ (due to the normalization of the Dirac operator).
The idea of studying non-unital spectral triples of finite volume for the Moyal
plane was originally proposed in \cite{GW11}.
By analogy with quantum mechanics, we call  $T_W$ the \emph{Wigner spectral triple},
and $T_S$ the \emph{Schr{\"o}dinger spectral triple}.

\subsection{On the relation between the Moyal plane and the $1$-point space}\label{sec:4.3}

In this section we show that the Schr{\"o}dinger spectral triple
of previous section is isomorphic to any of its truncations $(M_n(\C),\HH_n,D_n)$, with $n\geq 2$ (meaning
that they are related by an invertible correspondence).

It is well known that $\ell^2(\N)\otimes \overline{\C^n}$ is a Morita equivalence bimodule between $M_n(\C)$ and the $C^*$-algebra $\K$ of compact operators on $\ell^2(\N)$.
Here $\overline{\C^n}$ denotes row vectors, and the right action of $M_n(\C)$ is given by row-by-column multiplication.

In \S\ref{sec:mor}, we discussed correspondences $(\A,\HH,D) \xrightarrow{(\E,\nabla,U)} (\A',\HH',D')$ of finite dimensional spectral triples, and studied the corresponding category. Let us adapt the construction to the infinite-dimensional case. Clearly, $\E$ should be a Morita-equivalence bimodule between the $C^*$-completions of $\A$ and $\A'$, and $\nabla$ should be densely defined and its domain related to the domains of $D$ and $D'$.

A general theory is beyond the scope of this paper (one can see e.g.~\cite{Mes09}): we discuss here only the case of our interest.
A connection on a right Hilbert $A'$-module $\E'$ (with $A'$ the $C^*$-completion of $\A'$) will be then a map with dense domain $\E'_1\subset\E'$
and image in $\E'_1\otimes_{\A'}\Omega^1_{D'}(\A')$.
If $\HH$ is finite-dimensional, $\E_1\subset\E$ is the domain of the connection $\nabla$ and $\HH'_1\subset\HH'$ the domain of self-adjointness of $D'$,
a natural request is that $U$ maps $\E_1\otimes_{\A}\HH$ to $\HH'_1$, and $\mathrm{Ad}_U$ maps $\mathrm{End}^0_{\A}(\E)$ into the norm-closure of $\A'$.

If $\A'$ is non-unital, a further consideration is necessary:
the connection \eqref{eq:nablaD} is a map $\nabla_{D'}:\A'^+\to\A'^+\otimes_{\A'}\Omega^1_{D'}(\A')$,
where $\A'^+=\A'\oplus\C$ is the minimal unitization of $\A'$,
so $(A'^+,\nabla_{D'},m)$ might be a natural candidate for the identity morphism.
A counterexample is given by Moyal plane, where $\A'=\mc{S}(\N^2)$ has completion $A'=\K$,
and the neutral element for the tensor product of Morita-equivalence bimodules is $\K$, and not $\K^+$.

\begin{rem}
Recall that $\Omega^1_{D_n}(\A_n)$ is the rank $2$ free $M_n(\C)$-module with basis
given by the two elements \eqref{eq:omegauno}. A similar statement holds for Moyal.
Let us denote by the same symbol $e_{ij}$ the infinite-dimensional matrix with $1$ in position $(i,j)$ and zero everywhere else, $i,j\geq 0$. From
$$
\sqrt{\frac{\theta}{2}}\,e_{j0}\!\otimes\! 1_2\,[D,e_{1k}\!\otimes\! 1_2]=\mat{0 & e_{jk} \\ 0 & 0} \;,\qquad
\sqrt{\frac{\theta}{2}}\,e_{j1}\!\otimes\! 1_2\,[D,e_{0k}\!\otimes\! 1_2]=\mat{0 & 0 \\  e_{jk} & 0} \;,
$$
it follows that $\Omega_D(\mc{S}(\N^2))$ is the subspace of elements of $\mc{S}(\N^2)\otimes M_2(\C)$
with zeros on the blocks on the main diagonal.
\end{rem}

\begin{prop}\label{prop:4.2}
A correspondence:
$$
(\A_n:=M_n(\C),\HH_n,D_n)\xrightarrow{\;(\E,\nabla,U)\;} (\mc{S}(\N^2),\ell^2(\N)\otimes\C^2,D)
$$
is given as follows.
Let $\E:=\ell^2(\N)\otimes\overline{\C^n}$,
$\E_1:=\mc{S}(\N)\otimes \overline{\C^n}$,
and
$$
U:\ell^2(\N)\otimes\overline{\C^n}\otimes_{M_n(\C)}\C^n\otimes\C^2\to\ell^2(\N)\otimes\C^2
$$
be given by $(m\otimes\id)(\id\otimes m\otimes \id)$, with $m$ the multiplication map.
The connection is
\begin{gather}
\nabla :\E_1\to\E_1\otimes_{M_n(\C)}\Omega^1_{D_n}(M_n(\C))=\mat{0 & \E_1\otimes_{\A_n}\A_n \\ \E_1\otimes_{\A_n}\A_n & 0}
\notag \\[3pt]
\nabla\eta :=  \sqrt{\frac{2}{\theta}}\mat{0 & (\mf{a}^\dag\eta-\eta X_n^*)\otimes_{\A_n}1_n \\ (\mf{a}\eta-\eta X_n)\otimes_{\A_n}1_n & 0} \;.\label{eq:nabla}
\end{gather}
Here $\mf{a}$ and $\mf{a}^\dag$ act on the first factor of $\E_1$.
\end{prop}
\begin{proof}
The Leibniz rule for $\nabla$ is easy to check. From \eqref{eq:nabla}, one gets
$$
\nabla(\eta a) = \nabla(\eta) a+\sqrt{\frac{2}{\theta}}\mat{0 & \eta\otimes_{\A_n}[X_n^*,a] \\ \eta\otimes_{\A_n}[X_n,a] & 0}=\nabla(\eta) a+\eta\otimes_{\A_n}[D_n,a] \;.
$$
for any $\eta\in\E_1$ and $a\in\A_n=M_n(\C)$. Furthermore, by construction $U(1\otimes_\nabla D_n)=DU$.
\end{proof}

If $A$ is a unital $C^*$-algebra and $E$ a right Hilbert $A$-module, we can identify $E\otimes_AA$ with 
$E$. In particular, to simplify the notations, from now on we drop the several ``$\otimes_{\A_n}\A_n$''
and ``$\otimes_{\A_n}1_n$'', for example in \eqref{eq:nabla}.

For a non-unital algebra $A$ and a right $A$-module $E$, there is in general no isomorphism
$E\otimes_AA\to E$. For example, if $A=E$ is the ideal in $\C[x]$ generated by $x$, then
$E\otimes_AA$ is isomorphic to the ideal in $\C[x]$ generated by $x^2$. On the other hand, if
$A=\K$ and $E$ is the module $\overline\E=\C^n\otimes\overline{\ell^2(\N)}$, the multiplication map
$\E\otimes_{\K}\K\to\E$ has inverse defined on a basis by $e_i\otimes\bra{k}\mapsto
(e_i\otimes\bra{k})\otimes_{\K}e_{kk}$. In the following, we will identify $\overline\E\otimes_{\K}\K$
and $\overline\E$.

Note that the map $m:\overline{\ell^2(\N)}\otimes_{\K}\ell^2(\N)\to\C$ given by the inner product (row-by-column multiplication)
is invertible with inverse $m^*:1\mapsto\eta^*\otimes_{\K}\eta$ for any unit vector $\eta$.\footnote{For any two unit vectors $\xi$
and $\eta$, since $p:=\xi\inner{\eta,\,.\,}$ is a compact operator, one has $\eta^*\otimes_{\K}\eta=\xi^*p\otimes_{\K}\eta=
\xi^*\otimes_{\K}p\eta=\xi^*\otimes_{\K}\xi$.}

\begin{prop}\label{prop:4.3}
A correspondence
$$
(\A:=\mc{S}(\N^2),\ell^2(\N)\otimes\C^2,D)\xrightarrow{\;(\overline\E,\overline\nabla,\overline U)\;}  (M_n(\C),\HH_n,D_n)
$$
is given by the complex conjugate bimodule $\overline\E=\C^n\otimes\overline{\ell^2(\N)}$ of $\E$, with connection
\begin{gather}
\overline\nabla :\overline{\E}_1\to
\mat{0 & \overline{\E}_1 \\ \overline{\E}_1 & 0}
\;,\qquad\quad
\overline\nabla\xi :=\sqrt{\frac{2}{\theta}} \mat{0 & X_n^*\xi-\xi\mf{a}^\dag \\ X_n\xi-\xi\mf{a} & 0}  \;. \label{eq:nablaprime}
\end{gather}
Here we think of $\overline{\ell^2(\N)}$ and $\overline{\mc{S}(\N)}$ as row vectors and,
for $\xi\in\overline{\E}_1=\C^n\otimes\overline{\mc{S}(\N)}$, the operators $\mf{a},\mf{a}^\dag$ multiply from the right on the second factor, and $X_n,X_n^*\in M_n(\C)$ from the left on the first factor. The unitary $\overline U:\C^n\otimes\overline{\ell^2(\N)}\otimes_{\K}\ell^2(\N)\otimes\C^2\to\HH_n$ given by
$(m\otimes\id)(\id\otimes\inner{\,,\,}\otimes\id)$.
\end{prop}

\begin{proof}
The proof is analogous to the one of Prop.~\ref{prop:4.2}.
For all $b\in\A$:
$$
\overline\nabla(\xi b)-(\overline\nabla\xi)b=
\sqrt{\frac{2}{\theta}}\mat{0 & \xi[\mf{a}^\dag,b] \\ \xi[\mf{a},b] & 0}
=\sqrt{\frac{2}{\theta}}\mat{0 & \xi\,D_+(b) \\ \xi\,D_-(b) & 0}=\xi\,[D,b] \;,
$$
where we used the fact that $[\mf{a}^\dag,b]=D_+(b)$ and $[\mf{a},b]=D_-(b)$ belongs to $\A$.
With a straightforward computation one checks that $\overline U(1\otimes_{\overline\nabla}D)=D_n\overline U$.
\end{proof}

To prove that the morphisms in Prop.~\ref{prop:4.2} and \ref{prop:4.3} are one the inverse of the other,
we need to understand what is the identity morphism for Moyal. Since $\K$ is the neutral element for the
tensor product of Hilbert modules, the identity morphism must be represented by a correspondence
$(\K,\nabla_D,m)$, where $\nabla_D$ is a suitable connection densely defined on $\K$.
Since $\mc{S}(\N^2)\otimes_{\mc{S}(\N^2)}\Omega^1_D(\mc{S}(\N^2))\simeq\mat{0 & \mc{S}(\N^2) \\ \mc{S}(\N^2) & 0}$, we
define $\nabla_D:\mc{S}(\N^2)\to\mc{S}(\N^2)\otimes_{\mc{S}(\N^2)}\Omega^1_D(\mc{S}(\N^2))$ as the composition
of $d_D$ with the isomorphism above.

\begin{prop}
The morphisms in Prop.~\ref{prop:4.2} and \ref{prop:4.3} are one the inverse of the other.
More precisely, the composition in one order gives $(\A_n,\nabla_{D_n},m)$ (modulo similarity),
and in the other order gives $(\K,\nabla_D,m)$.
\end{prop}

\begin{proof}
Using the notations of Prop.~\ref{prop:4.2} and \ref{prop:4.3}, the composition of correspondences
in one order gives
$$
\E':=\overline\E\otimes_{\K}\E \;,\qquad
\nabla':=\{(\id\otimes\mathrm{Ad}_{U^*})\overline\nabla\}\odot\nabla \;,\qquad
U':=\overline U(\id\otimes U) \;,
$$
while in the other order gives
$$
\E'':=\E\otimes_{\A_n}\overline\E \;,\qquad
\nabla'':=\{(\id\otimes\mathrm{Ad}_{\overline U^*})\nabla\}\odot\overline\nabla \;,\qquad
U'':=U(\id\otimes\overline U) \;,
$$
where in $\E''$ the completed tensor product is understood. An $\A_n$-bimodule map
$$
V':\E'=\C^n\otimes\overline{\ell^2(\N)}\otimes_{\K}\ell^2(\N)\otimes\overline{\C^n}\to M_n(\C)=\A_n
$$
is given by row-by-column multiplication.
By construction $m(V'\otimes_{\A_n}\id_{\HH_n})=U'$,
establishing the equivalence between $(\E',\nabla',U')$ and $(\A_n,\nabla_{D_n},m)$.

Similarly, row-by-column multiplication gives a unitary $\K$-bimodule map:
$$
V'':\ell^2(\N)\otimes\overline{\C^n}\otimes_{M_n(\C)}\C^n\otimes\overline{\ell^2(\N)}\to\mc{L}^2(\ell^2(\N)) \;.
$$
By completion, we get a unitary map $V'':\E''\to\K$, which is a similarity between $(\E'',\nabla'',U'')$ and $(\K,\nabla_D,m)$.
\end{proof}

We can conclude that the Schr{\"o}dinger spectral triple is ``isomorphic'' to any of its truncations,
although not metrically equivalent. On the other hand, it is metrically equivalent to
the Wigner spectral triple, but not isomorphic. The situation is summarized in Fig.~\ref{fig}.

\medskip

\begin{figure}[th]
\begin{scriptsize}
\begin{tikzpicture}[
    scale=0.75,
    grow=right,
    level 1/.style={sibling distance=6.6cm,level distance=9.2cm},
    edge from parent/.style={very thick,draw=blue!40!black!60,
        shorten >=5pt, shorten <=5pt},
    edge from parent path={(\tikzparentnode.east) -- (\tikzchildnode.west)},
    kant/.style={text width=6cm, text centered, sloped},
    every node/.style={text ragged, inner sep=2mm},
    punkt/.style={rectangle, rounded corners, shade, top color=white,
    bottom color=white,%blue!50!black!20,
    draw=blue!40!black!60,
    very thick }
    ]

\node[punkt, text width=5.5em] { \ One point space $I_n/\mathbb{Z}_n$}
    child {
        node[punkt] [rectangle split, rectangle split, rectangle split parts=3,
         text ragged] (A) {
            \textbf{Schr{\"o}dinger spectral triple}
                  \nodepart{second}
            Irreducible representation
                  \nodepart{third}
            Finite volume $V=2\theta$
        }
        edge from parent
            node[kant, below, pos=.5] { Same smooth structure. \\[2pt] Different metrics. }
    }
    child {
        node[punkt] [rectangle split, rectangle split, rectangle split parts=3,
         text ragged] (B) {
            \textbf{Wigner spectral triple}
                  \nodepart{second}
            GNS representation
                  \nodepart{third}
            Infinite volume
        }
        edge from parent
            node[kant, above, pos=.5] { Different smooth structures. \\[2pt] Different metrics. }
    };

\draw[edge from parent] (A)-- node[kant,below] { Different smooth \\ structures. \\[2pt] Same metric. } (B) ;

\coordinate (C) at ($(A)!5.1cm!270:(B)$);
\coordinate (D) at ($(A)!7cm!270:(B)$);

\coordinate (E) at ($(B)!5.3cm!90:(A)$);
\coordinate (F) at ($(B)!7cm!90:(A)$);

\draw[edge from parent,->,shorten >=10pt] (A) -- (C) node[kant, above, pos=.5] {\mbox{$\hspace{-7pt}\theta=0$}};

\draw (D) node[punkt] {No commutative limit};

\draw[edge from parent,->,shorten >=5pt] (B) -- (E) node[kant, above, pos=.5] {\mbox{$\theta=0$}};

\draw (F) node[punkt,rectangle split, rectangle split, rectangle split parts=3,
         text ragged] {
            \textbf{Euclidean plane}
                  \nodepart{second}
            Canonical S.T.
                  \nodepart{third}
            Infinite volume
            };

\end{tikzpicture}
\end{scriptsize}
\vspace{-3mm}
\caption{Entity-relationship diagram.}
\label{fig}
\end{figure}
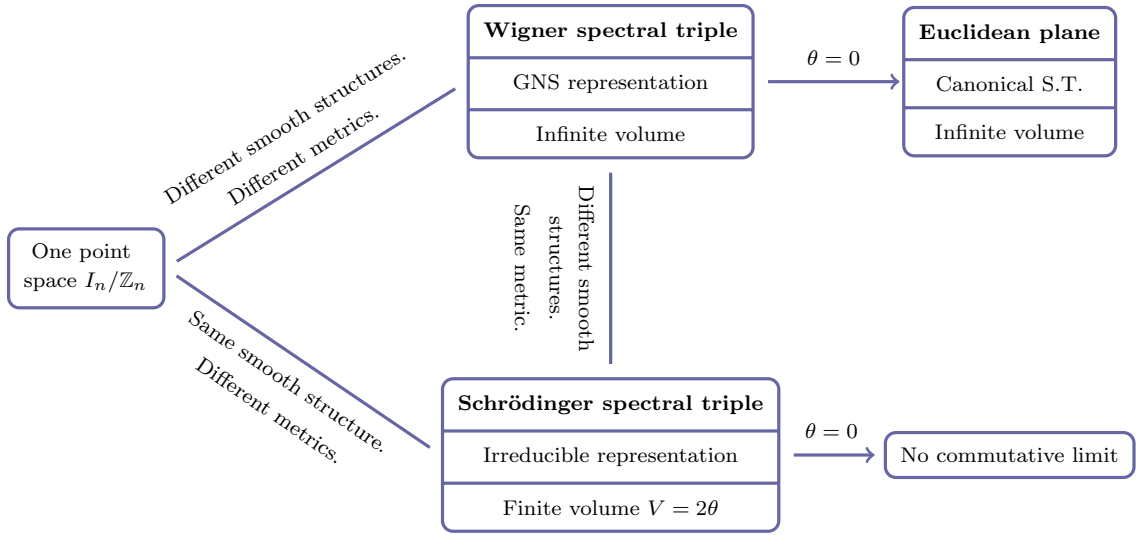

\section{Metric aspects of the Moyal plane}\label{sec:4.4}
In this section we provide a short proof of the formula for the distance between coherent states \cite{CDMW09,MT11,DLM13}, using entirely the matrix notation.
We then show that states corresponding to basis vectors of $\ell^2(\N)$, with the spectral distance, form a metric space that converge to
$\R^+_0$ for $\theta\to 0$.

\subsection{Euclidean planes inside Moyal plane: generalized coherent states}\label{sec:4.4.1}
Coherent states are the states associated to the vectors
$$
\ket{z}=e^{-\frac{1}{4\theta}|z|^2}\sum_{n=0}^\infty \frac{1}{\sqrt{n!}}\left(\frac{z}{\sqrt{2\theta}}\right)^n\ket{n} \;,
$$
with $z\in\C$. They can be generalized as follows.	

For $z\in\C$ and $\alpha\in[0,2\pi[$, consider the unitary operators
$$
T(z)=\exp\big\{\tfrac{1}{\sqrt{2\theta}}(z\mf{a}^\dag-\bar z\mf{a})\big\} \;,\qquad
R(\alpha)=e^{i\alpha\,\mf{a}^\dag\mf{a}} \;,
$$
giving a projective representation of the Galilean group $G:=SO(2)\ltimes\R^2$. More precisely,
they give a representation of the semidirect product $SO(2)\ltimes H_3(\R)$,
where $H_3(\R)$ is the Heisenberg group:
$$
T(z)T(w)=e^{\frac{i}{2\theta}\Im(z\bar w)}T(z+w) \;,\quad
R(\alpha)R(\beta)=R(\alpha+\beta) \;,\quad
R(\alpha)T(z)=T(e^{i\alpha}z)R(\alpha) \;.
$$
In fact, we can replace $\alpha$ by any complex number $\tau$ with $\Im(\tau)\geq 0$, and the above
relations are still valid, but $R(\tau)$ is no longer unitary: for $\Im(\tau)>0$ it is a compact operator (in fact,
of rapid decay) and has no bounded inverse. The norm of $R(\tau)$ is still $1$ for any $\tau$.

Coherent states can be generalized by fixing a ground state $\ket{\psi_0}$ (normalized to $1$), not
necessarily $\ket{0}$, and defining
$$
\ket{\psi_z}=T(z)\ket{\psi_0} \;.
$$
The corresponding state will be denoted $\Psi_z$:
$$
\Psi_z(a)=\inner{\psi_z|a|\psi_z} \;.
$$
Note that they transform under $G$ according to the law:
$$
T(w)\ket{\psi_z}=e^{\frac{i}{2\theta}\Im(z\bar w)}\ket{\psi_{z+w}} \;,\qquad
R(\alpha)\ket{\psi_z}=\ket{\psi_{e^{i\alpha}z}} \;.
$$
When considering the corresponding states, the phase factor $e^{\frac{i}{2\theta}\Im(z\bar w)}$
simplifies and the projective representation of $G$ becomes an actual representation:
$$
\mathrm{Ad}^*_{T(w)}\Psi_z=\Psi_{z+w} \;,\qquad
\mathrm{Ad}^*_{R(\alpha)}\Psi_z=\Psi_{e^{i\alpha}z} \;.
$$
Here for $U$ a unitary operator on $\ell^2(\N)$,
$\mathrm{Ad}^*_U$ denotes the pull-back to normal states of the adjoint representation:
$$
\mathrm{Ad}^*_U\inner{\psi|a|\psi}=\inner{\psi|U^*aU|\psi} \;,\qquad\forall\;\psi\in\HH,\;a\in\B(\HH)\,.
$$

\begin{lemma}\label{lemma:three}
Let $N\geq 1$ and
\begin{equation}\label{eq:aN}
a_N=\mf{a}^\dag R(\tfrac{i}{N})+R(\tfrac{i}{N})\,\mf{a} \;.
\end{equation}
Then, the element
$$
b_N=\frac{a_N}{1+(e^{\frac{1}{N}}-1)N}
$$
satisfies $\|[D,b_N]\|\leq 1$, with $D$ as in \eqref{eq:Dirred}.
\end{lemma}
\begin{proof}
Firstly, $R(\tfrac{i}{N})\,\mf{a}$ is a rapid decay matrix, hence $a_N=a_N^*$ belongs to $\A$.

On the other hand, for $N\geq 1$ 
$$
R(\tfrac iN)\ket{n}=e^{-\frac{n}{N}}\ket{n}
$$ 
Since $[\mf{a}^\dag,R(\tfrac{i}{N})]=(e^{\frac{1}{N}}-1)R(\tfrac{i}{N})\mf{a}^\dag$, by the Leibniz rule we have
$$
[\mf{a}^\dag,a_N]=(e^{\frac{1}{N}}-1)(\mf{a}^\dag R(\tfrac{i}{N})\mf{a}^\dag+R(\tfrac{i}{N})\mf{a}^\dag\mf{a})-R(\tfrac{i}{N}) \;.
$$
The function $f(x)=xe^{-\frac{x}{N}}$ has its maximum at $x=N$, where it is equal to $f(N)=Ne^{-1}$. Then
\begin{align*}
\|\mf{a}^\dag R(\tfrac{i}{N})\mf{a}^\dag\| &=\sup_n\big\{e^{-\frac{n+1}{N}}\sqrt{(n+1)(n+2)}\big\}\leq e^{\frac{1}{N}}
\sup_n\big\{(n+2)e^{-\frac{n+2}{N}}\big\}=Ne^{\frac{1}{N}-1}\leq N \;,
\\
\|R(\tfrac{i}{N})\mf{a}^\dag\mf{a}\| &=\sup_n\big\{ne^{-\frac{n}{N}}\big\}=Ne^{-1}\leq N \;.
\end{align*}
Hence
$$
\|[\mf{a}^\dag,a_N]\|\leq \|R(\tfrac{i}{N})\|+(e^{\frac{1}{N}}-1)(\|\mf{a}^\dag R(\tfrac{i}{N})\mf{a}^\dag\|+\|R(\tfrac{i}{N})\mf{a}^\dag\mf{a}\|)
\leq 1+(e^{\frac{1}{N}}-1)N \;.
$$
By Lemma \ref{lemma:1}, this concludes the proof.
\end{proof}

\begin{lemma}\label{lemma:four}
Let $a_N$ be the element in \eqref{eq:aN}. Then
$$
\lim_{N\to\infty}\big\{
\Psi_r(a_N)-\Psi_0(a_N)
\big\}=\sqrt{\tfrac{2}{\theta}}\,r \;,
$$
for any $r>0$.
\end{lemma}
\begin{proof}
Since
$$
[\mf{a},T(r)]=[\mf{a}^\dag,T(r)]=\tfrac{r}{\sqrt{2\theta}}\,T(r) \;,\qquad
R(\tfrac{i}{N})T(r)=T(e^{-\frac{1}{N}}r)R(\tfrac{i}{N}) \;,
$$
we have
$$
a_NT(r)=T(e^{-\frac{1}{N}}r)\big\{a_N+\tfrac{r}{\sqrt{2\theta}}(e^{-\frac{1}{N}}+1)R(\tfrac{i}{N})\big\} \;.
$$
Since $T(r)$ is a strongly continuous one-parameter group of unitary transformations (by Stone's theorem),
$T(e^{-\frac{1}{N}}r)$ is norm convergent to $T(r)$ for $N\to\infty$ and
$$
\lim_{N\to\infty}\big\{
\Psi_r(a_N)-\Psi_0(a_N)
\big\}=\sqrt{\tfrac{2}{\theta}}\,r\,\lim_{N\to\infty}\Psi_0\big(R(\tfrac{i}{N})\big) \;.
$$
If $\ket{\psi_0}=\sum_{n\geq 0}c_n\ket{n}$ with $\sum_{n\geq 0}|c_n|^2=1$, one has
$$
\Psi_0\big(R(\tfrac{i}{N})\big)=\sum\nolimits_{n\geq 0}|c_n|^2e^{-\frac{n}{N}} \;.
$$
By Weierstrass M-test, the series is uniformly convergent to a continuous functions of $\frac{1}{N}$,
and $\lim_{N\to\infty}\Psi_0\big(R(\tfrac{i}{N})\big)=\sum\nolimits_{n\geq 0}|c_n|^2\lim_{N\to\infty}e^{-\frac{n}{N}}=1$.
\end{proof}

\begin{prop}[\cite{MT11}]
For any $z,z'\in\C$, we have $d_D(\Psi_z,\Psi_{z'})=|z-z'|$.
\end{prop}

\begin{proof}
Suppose $U$ is a unitary transformation of $\ell^2(\N)$. If
$\|[\mf{a}^\dag,\mathrm{Ad}_U(a)]\|=\|[\mf{a}^\dag,a]\|$ for all $a\in\A$,
then the distance is invariant under $\mathrm{Ad}_U^*$.
For $U=R(\alpha)$ one has $U\mf{a}^\dag U^*=e^{i\alpha}\mf{a}^\dag$ and the above condition is satisfied.
For $U=T(z)$ one has $U\mf{a}^\dag U^*=\mf{a}^\dag-\frac{\bar z}{\sqrt{2\theta}}$ and the above condition is satisfied.
Hence the distance is $G$-invariant, $d_D(\Psi_z,\Psi_{z'})=d_D(\Psi_0,\Psi_{|z-z'|})$, and it
is enough to prove the proposition when $z=0$ and $z'=r>0$.

From Lemma \ref{lemma:three} and \ref{lemma:four} it follows that
$$
d_D(\Psi_0,\Psi_r)\geq \sqrt{\tfrac{\theta}{2}}\lim_{N\to\infty}\big\{
\Psi_r(a_N)-\Psi_0(a_N)\big\}=r \;.
$$
On the other hand, if $A$ is a self-adjoint operator, $U(t)=e^{-itA}$ the corresponding one-parameter unitary group
and $\varphi$ a state, one has
$$
\frac{\de}{\de t}\varphi(e^{itA}ae^{-itA})=i\varphi(e^{itA}[A,a]e^{-itA}) \;.
$$
For $\varphi=\Psi_0$ and $U(t)=T(t)$, whose generator is $A=\frac{i}{\sqrt{2\theta}}(\mf{a}^\dag-\mf{a})$, we get
$$
\Psi_t([\mf{a}-\mf{a}^\dag,a])=\sqrt{2\theta}\frac{\de}{\de t}\Psi_t(a) \;.
$$
For any $a\in\mathfrak{B}$, by integrating previous equation we get:
$$
\Psi_r(a)-\Psi_0(a)=\frac{1}{\sqrt{2\theta}}\int_0^r\Psi_t([\mf{a},a]-[\mf{a}^\dag,a])\de t
\leq \frac{r}{\sqrt{2\theta}}\big(\|[\mf{a},a]\|+\|[\mf{a}^\dag,a]\|\big)\leq \frac{2r}{\sqrt{2\theta}}
 \;.
$$
This proves that $d_D(\Psi_0,\Psi_r)\leq r$.
\end{proof}

We conclude this short section by giving a geometric interpretation of the above result.
Let $f_{mn}$ be the matrix basis of $L^2(\R^2)$ and $\psi_z$ the unit vector
$$
\psi_z=\frac{e^{-\frac{1}{2\theta}|z|^2}}{\sqrt{2\pi\theta}}\sum_{m,n\geq 0}\frac{1}{\sqrt{m!\hspace{1pt}n!}}
\left(\frac{z}{\sqrt{2\theta}}\right)^m
\left(\frac{\bar z}{\sqrt{2\theta}}\right)^n
f_{mn} \;.
$$
Explicitly, as a function $\psi_z(\xi)=(\pi\theta)^{-\frac{1}{2}}e^{-\frac{1}{2\theta}|\xi-z|^2}$ is a Gaussian.

This defines a state on Moyal algebra that coincides with the coherent state $\Psi_z$:
$$
\inner{\psi_z,a\ast_\theta\psi_z}=
e^{-\frac{1}{2\theta}|z|^2}\sum_{m,n\geq 0}\frac{1}{\sqrt{m!\hspace{1pt}n!}}
\left(\frac{\bar z}{\sqrt{2\theta}}\right)^m
\left(\frac{z}{\sqrt{2\theta}}\right)^n
\,a_{mn}
\equiv \Psi_z(a) \;.
$$
But it also defines a Gaussian state on $\mc{S}(\R^2)$ given by
$\inner{\psi_z,f\psi_z}=\int_{\R^2}|\psi_z(x)|^2f(x)\de^2x$.
The Wasserstein distance between two Gaussian states with the same variance
is the Euclidean distance between the peaks (see e.g.~\S3.2 of \cite{DM09}).
Our computation proves that the distance is undeformed after quantization.

\subsection{Distance between eigenstates of the harmonic oscillator}\label{sec:4.4.2}

Let $\Psi^\theta_m$ be the state associated to the $m$-th basis vector ($m$-th eigenvector of the quantum harmonic oscillator):
$$
\Psi^\theta_m(a)=a_{mm}
$$
These are rotationally invariant: since $R(\alpha)\ket{n}=e^{i\alpha}\ket{n}$ implies $\mathrm{Ad}^*_{R(\alpha)}\Psi^\theta_m=\Psi^\theta_m$ for all $\alpha\in\R$.
It will be clear later why we indicate explicitly the deformation parameter $\theta$.

It is shown in \cite{CDMW09} that for all 
$m<n$, one has
$$
d_D(\Psi^\theta_m,\Psi^\theta_n)=\sqrt{\frac{\theta}{2}}\sum_{k=m+1}^n\frac{1}{\sqrt{k}} \;.
$$
Let $X_0$ be the metric space given by the set $\R^+_0$ with Euclidean distance.
For $\theta>0$, the states $\Psi^\theta_m$ with distance $d_D$
form a metric space $X_\theta$, which is isometrically embedded into
the Euclidean half-line by 
\begin{equation}\label{eq:xthetam}
f_\theta:X_\theta\to X_0, \quad   f_\theta(\Psi^\theta_m)=x_m^\theta:=\sum_{k=1}^m \sqrt{\frac{\theta}{2k}} \;,
\end{equation}
for $m\geq 1$ and and $x_0^\theta=0$.
Indeed, one easily checks that:
$$
d_D(\Psi_m^\theta,\Psi_n^\theta)=|x_m^\theta-x_n^\theta| \;.
$$

Intuitively, we expect that the large scale structure of Moyal plane is the one of an ordinary
Euclidean plane, that is for $\theta\to 0^+$ the Moyal plane approximates metrically
the Euclidean plane. A rigorous way to formulate this is by showing
that the Moyal plane is convergent to
the Euclidean plane in the Gromov-Hausdorff distance.
Here we prove a smaller result, i.e.~that for $\theta\to 0$ the metric spaces $X_\theta$ converge
to the half line for the Gromov-Hausdorff distance \cite{BBI01}.

\begin{prop}
$X_\theta\to X_0$ for the Gromov-Hausdorff distance.
\end{prop}
\begin{proof}
Let $Y_\theta=\{x^\theta_m\}_{m\in\N}$ with $x^\theta_m$ given by \eqref{eq:xthetam}. These
are metric subspaces of $X_0$ (with Euclidean distance), and being isometric to $X_\theta$
it is enough to show that $Y_\theta\to X_0$ for the Hausdorff distance.
For this, we need to find, for any $\xi\in\R^+_0$, a set $\{\xi_\theta\in Y_\theta\}_{\theta>0}$
such that $|\xi-\xi_\theta|\to 0$ for $\theta\to 0^+$ \cite[pag.~253]{BBI01}.

As a preliminary step, notice that for any $m<n$ we have
$$
x_n^\theta=\sum_{k=1}^n \sqrt{\frac{\theta}{2k}}
\geq\int_{1}^{n+1}\sqrt{\frac{\theta}{2k}}\,\de k
=\sqrt{2\theta}\big(\sqrt{n+1}-1\big)\geq\sqrt{2\theta n}-\sqrt{2\theta} \;,
$$
and
$$
x_n^\theta\leq\int_0^n
\sqrt{\frac{\theta}{2k}}\,\de k=\sqrt{2\theta n}
\;.
$$
Let $\xi\in\R^+_0$ be an arbitrary point. For any $\theta>0$, $\xi^2$ can be written (in a unique way) as
$\xi^2=2\theta n_\theta+\epsilon_\theta$, with $n_\theta$ integer and $0\leq\epsilon_\theta<2\theta$.
Note that $\lim_{\theta\to 0}\epsilon_\theta=0$ and so $\lim_{\theta\to 0}\sqrt{2\theta n_\theta}=\xi$.
If we set $\xi_\theta=x_{n_\theta}^\theta$, then from the above discussion we have
$$
\xi-\sqrt{2\theta n_{\theta}} \leq \xi-\xi_\theta \leq \xi-\sqrt{2\theta n_{\theta}}+\sqrt{2\theta} \;,
$$
which proves that $\lim_{\theta\to 0^+}|\xi-\xi_\theta|=0$.
\end{proof}

%%% ======================================================================

\subsection*{Acknowledgments}
We thank Bram Mesland and Walter van Suijlekom for discussions and correspondence, and an anonymous referee for many valuable remarks and suggestions.
This research was partially supported by UniNA and Compagnia di San Paolo under the program STAR 2013.
F.L.~is partially supported by CUR Generalitat de Catalunya under project FPA2010-20807.


\medskip

\begin{thebibliography}{99}
\itemsep=-2pt
\small

\bibitem{Wei97}
S. Bates and A. Weinstein, \textit{Lectures on the Geometry of Quantization},
Berkeley Math. Lect. Notes, AMS, 1997.

\bibitem{BCL12}
P. Bertozzini, R. Conti and W. Lewkeeratiyutkul, \textit{Categorical Non-commutative Geometry},
J. Phys.: Conf. Ser. 346 (2012), 012003.

\bibitem{BBI01}
D. Burago, Y. Burago and S. Ivanov, \textit{A Course in Metric Geometry},
  Graduate Studies in Math. 33, AMS, 2001.
  
 \bibitem{CDMW09}
E.~Cagnache, F.~D'Andrea, P.~Martinetti and J.-C. Wallet, \textit{The spectral
  distance on the Moyal plane}, J. Geom. Phys. 61 (2011), 1881--1897.

\bibitem{CCvS13}
A.H. Chamseddine, A. Connes and W.D. van Suijlekom, \textit{Inner fluctuations in noncommutative geometry without
  the first order condition}, J. Geom. Phys. 73 (2013), 222--234.

\bibitem{Con94}
A.~Connes, \textit{Noncommutative Geometry}, Academic Press, 1994.

\bibitem{CC06}
A. Connes and A.H. Chamseddine, \textit{Inner fluctuations of the spectral action},
  J. Geom. Phys. 57 (2006), 1--21.

\bibitem{CDS98}
A. Connes, M.R. Douglas and A. Schwarz, \textit{Noncommutative Geometry and Matrix Theory: Compactification on Tori},
	JHEP 02 (1998), 003.

\bibitem{CM08}
A.~Connes and M.~Marcolli, \textit{Noncommutative geometry, quantum fields and
  motives}, Colloquium Publications, vol.~55, AMS, 2008.

\bibitem{DLM13}
F.~D'Andrea, F.~Lizzi and P.~Martinetti, \textit{Spectral geometry with a cut-off: topological and metric aspects}, 
J. Geom. Phys. 82 (2014), 18--45.

\bibitem{DM09}
F.~D'Andrea and P.~Martinetti, \textit{A view on Transport Theory from
  Noncommutative Geometry}, SIGMA 6 (2010), 057.

\bibitem{GGISV04}
V.~Gayral, J.M.~Gracia-Bond{\'\i}a, B.~Iochum, T.~Sch{\"u}cker, J.C. V{\'a}rilly,
\textit{Moyal planes are spectral triples}, Commun. Math. Phys. 246 (2004), 569--623.

\bibitem{GW11}
V. Gayral and R. Wulkenhaar,
\textit{Spectral geometry of the Moyal plane with harmonic propagation},
J. Noncommut. Geom. 7 (2013), 939--979.

\bibitem{GVF01}
J.M. Gracia-Bond\'ia, J.C. V\'arilly and H. Figueroa,
\textit{Elements of Noncommutative Geometry},
Birkh\"auser, Boston, 2001.

\bibitem{Lan02}
G.~Landi, \textit{An introduction to noncommutative spaces and their
  geometries}, Springer, 2002.

\bibitem{lls98}
G.~Landi, F.~Lizzi and R.~J.~Szabo,
\textit{String geometry and the noncommutative torus},
Commun. Math. Phys. 206 (1999), 603.

\bibitem{Lan98book}
N.P.~Landsman, \textit{Topics Between Classical And Quantum Mechanics}, Springer, 1998.

\bibitem{Lod97}
J.L. Loday, \textit{Cyclic homology}, Springer-Verlag, 1997.

\bibitem{Mad95}
J. Madore, \textit{An introduction to noncommutative differential geometry and its physical applications},
  London Math. Soc. lect. note series 206, Cambridge Univ. Press, 1995.

\bibitem{Mar05}
M. Marcolli, \textit{Arithmetic Noncommutative Geometry},
  Univ. Lecture Series 36, AMS, 2005.

\bibitem{MT11}
P.~Martinetti and L.~Tomassini, \textit{Noncommutative geometry of the
Moyal plane: translation isometries, Connes' distance on coherent states, Pythagoras equality},
Commun. Math. Phys. 323 (2013), 107--141.

\bibitem{Mes09}
B.~Mesland, \textit{Unbounded bivariant $K$-theory and correspondences in noncommutative geometry},
	J. Reine Angew. Math. 691 (2014), 101--172.
	
	, DOI:10.1515/crelle-2012-0076 (on-line, 2012).

\bibitem{Rie74}
M.A. Rieffel, \textit{Morita Equivalence for $C^*$-and $W^*$-algebras}, Journal of Pure
   and Applied Algebra 5 (1974), 51--96.

\bibitem{Rie81}
M.A. Rieffel, \textit{$C^*$-algebras associated with irrational rotations},
   Pacific J. Math. 93 (1981), 415--429.

\bibitem{Var06}
J.C. V\'arilly, \textit{An introduction to noncommutative geometry},
EMS Lect. Ser. in Math., 2006.

\bibitem{Ven11}
J.J. Venselaar, \textit{Morita ``equivalences'' of equivariant torus spectral triples},
Lett. Math. Phys. 103 (2013), 131--144.

\bibitem{Wil07}
D.P.~Williams, \textit{Crossed Products of $C^*$-Algebras},
Math. Surveys and Mono. 134, AMS, 2007.

\end{thebibliography}
\end{document}